\documentclass[11pt]{article}
\usepackage{authblk}
\usepackage{microtype}
\usepackage{amsmath,amsthm,amssymb}
\usepackage{hyperref}
\usepackage[top=2cm,bottom=2cm,left=3cm,right=3cm]{geometry}
\usepackage{xspace}
\usepackage[T1]{fontenc}
\usepackage{stmaryrd}
\usepackage{xifthen}
\usepackage{bbm}

\title{On all things star-free}

\author[1]{Thomas Place\thanks{Support from the DeLTA project (ANR-16-CE40-0007).}}
\author[2]{Marc Zeitoun$^*$}
\affil[1]{LaBRI, Bordeaux University and IUF, France}
\affil[2]{LaBRI, Bordeaux University, France}

\newcommand{\canec}{\ensuremath{\sim_\Cs}\xspace}
\newcommand{\caned}{\ensuremath{\sim_\Ds}\xspace}
\newcommand{\caneg}{\ensuremath{\sim_\Gs}\xspace}

\newcommand{\sclac}{\ensuremath{{A^*}/{\canec}}\xspace}
\newcommand{\dclac}{\ensuremath{{A^*}/{\caned}}\xspace}
\newcommand{\gclac}{\ensuremath{{A^*}/{\caneg}}\xspace}

\newcommand{\fo}{\ensuremath{\textup{FO}}\xspace}

\newcommand{\MOD}{\textup{MOD}\xspace}

\newcommand{\fowm}{\mbox{\ensuremath{\fo({<},\MOD)}}\xspace}

\newcommand{\md}{\ensuremath{\textup{\MOD}}\xspace}

\newcommand{\grp}{\ensuremath{\textup{GR}}\xspace}

\newcommand{\bsd}{\ensuremath{\textup{SD}}\xspace}

\newcommand{\sfp}[1]{\ensuremath{\mathit{SF}(#1)}\xspace}
\newcommand{\bsdp}[1]{\ensuremath{\mathit{SD}(#1)}\xspace}

\newcommand{\imprint}{imprint\xspace}
\newcommand{\imprints}{imprints\xspace}

\newcommand{\tame}{multiplicative\xspace}

\newcommand{\Ratms}{Rating maps\xspace}

\newcommand{\ratms}{rating maps\xspace}
\newcommand{\ratm}{rating map\xspace}

\newcommand{\Nice}{Nice\xspace}
\newcommand{\nice}{nice\xspace}

\newcommand{\mratm}{multiplicative rating map\xspace}
\newcommand{\mratms}{multiplicative rating maps\xspace}

\newcommand{\Mratms}{Multiplicative rating maps\xspace}

\newcommand{\iden}{\veps-approximation\xspace}
\newcommand{\idens}{\veps-approximations\xspace}

\newcommand{\prin}[2]{\ensuremath{\Is[#1](#2)}\xspace}

\newcommand{\opti}[2]{\ensuremath{\Is_{#1}\left[#2\right]}\xspace}
\newcommand{\copti}[1]{\opti{\Cs}{#1}}

\newcommand{\popti}[3]{\ensuremath{\Ps_{#1}^{#2}[#3]}\xspace}

\newcommand{\iopti}[2]{\ensuremath{\fri_{#1}[#2]}\xspace}
\newcommand{\ioptic}[1]{\iopti{\Cs}{#1}}

\newcommand{\sfcopti}{\opti{\sfp{\Cs}}{\rho}}
\newcommand{\csfcopti}{\popti{\sfp{\Cs}}{\Cs}{\rho}}

\newcommand{\typ}[2]{\ensuremath{[#1]_{#2}}\xspace}
\newcommand{\ctype}[1]{\typ{#1}{\Cs}}

\newcommand{\gtype}[1]{\typ{#1}{\Gs}}
\newcommand{\cmult}{\ensuremath{\mathbin{\scriptscriptstyle\bullet}}}

\usepackage{stmaryrd}

\newcommand{\Cs}{\ensuremath{\mathcal{C}}\xspace}
\newcommand{\Ds}{\ensuremath{\mathcal{D}}\xspace}

\newcommand{\Gs}{\ensuremath{\mathcal{G}}\xspace}

\newcommand{\Is}{\ensuremath{\mathcal{I}}\xspace}

\newcommand{\Ps}{\ensuremath{\mathcal{P}}\xspace}

\newcommand{\Fb}{\ensuremath{\mathbf{F}}\xspace}
\newcommand{\Gb}{\ensuremath{\mathbf{G}}\xspace}
\newcommand{\Hb}{\ensuremath{\mathbf{H}}\xspace}

\newcommand{\Kb}{\ensuremath{\mathbf{K}}\xspace}
\newcommand{\Lb}{\ensuremath{\mathbf{L}}\xspace}

\newcommand{\Ub}{\ensuremath{\mathbf{U}}\xspace}
\newcommand{\Vb}{\ensuremath{\mathbf{V}}\xspace}
\newcommand{\Wb}{\ensuremath{\mathbf{W}}\xspace}

\newcommand{\fri}{\ensuremath{\mathbbm{i}}\xspace}

\newcommand{\vari}{quotient-closed Boolean algebra\xspace}

\newcommand{\varis}{quotient-closed Boolean algebras\xspace}

\newcommand{\nat}{\ensuremath{\mathbb{N}}\xspace}
\def\inv{^{-1}}

\newcommand{\veps}{\ensuremath{\varepsilon}\xspace}

\newcommand{\dclosp}[1]{\ensuremath{\mathord{\downarrow_{#1}}}\xspace}
\newcommand{\dclosr}{\dclosp{R}}

\newcommand{\brataux}[2]{\ensuremath{\xi_{#1}[#2]}\xspace}

\newcommand{\bratauxd}{\brataux{\Ds}{\rho}}

\newcommand{\bratauxsfc}{\brataux{\sfp{\Cs}}{\rho}}

\theoremstyle{plain}
\newtheorem{theorem}{Theorem}

\newtheorem{fact}[theorem]{Fact}
\newtheorem{proposition}[theorem]{Proposition}
\newtheorem{lemma}[theorem]{Lemma}
\newtheorem{remark}[theorem]{Remark}
\newtheorem{example}[theorem]{Example}

\newcounter{sauvegarde}
\newcommand\adjustc[1]{
  \setcounter{sauvegarde}{\thetheorem}
  \setcounterref{theorem}{#1}
  \addtocounter{theorem}{-1}
}

\newcommand\restorec{
  \setcounter{theorem}{\thesauvegarde}
}

\begin{document}

\maketitle

\begin{abstract}
  We investigate the star-free closure, which associates to a class of languages its closure under~Boolean operations and marked concatenation. We prove that the star-free~closure of any finite class~and of any class of groups languages with decidable separation (plus mild additional properties) has decidable separation. We actually show decidability of a stronger property, called covering. This~generalizes many results on the subject in a unified framework. A key ingredient is that star-free closure coincides with another closure operator where Kleene stars are also allowed in restricted~contexts.
\end{abstract}

\section{Introduction}
\label{sec:intro}
This paper investigates a remarkable operation on classes of languages: the \emph{star-free closure}. It builds a new class \sfp{\Cs} from an input class \Cs by closing it under union, complement and concatenation. This generalizes an important specific class: the one of \emph{star-free languages}, \emph{i.e.}, the star-free closure of the class consisting of all finite languages. Star-free languages~are those that can be defined in first order logic~\cite{mnpfo}. The correspondence was lifted to the~quantifier alternation hierarchy of first order~logic by Thomas~\cite{ThomEqu}, which corresponds to a classification~of star-free languages: the dot-depth hierarchy~\cite{BrzoDot}. These results extend to the star-free closure~\cite{pzgenconcat}. For each input class \Cs, \sfp{\Cs} corresponds to a variant of first-order logic (specified by the set of predicates that are allowed). Moreover, its quantifier alternation hierarchy corresponds to a classification of \sfp{\Cs}: the concatenation hierarchy of basis \Cs.

Schützenberger proved that~one may decide whether a regular language is star-free~\cite{sfo}. This result established a framework for investigating and understanding classes of languages, based on the \emph{membership problem}: is it decidable to test whether an input regular language belongs to the class under investigation?
Similar results were obtained for other prominent classes. Yet, this fruitful line of research also includes some of the most famous open problems in automata theory. For example, only the first levels of the dot-depth hierarchy are known to have decidable membership (see~\cite{jep-dd45} for a survey).

Recently, these results were unified and generalized. First, the problem itself was strengthened: membership was replaced by separation as a means to investigate classes. The separation problem asks whether two input languages can be separated by one from the class under study. While more general and difficult than membership, separation is also more flexible. This was exploited to show that separation is decidable for several levels in the dot-depth hierarchy~\cite{pzqalt,pseps3}. In fact, this is a particular instance of a \emph{generic} result applying to every hierarchy whose basis \Cs is \emph{finite} and satisfies some mild properties~\cite{pseps3j,pzboolpol}. Moreover, the same result was obtained when the basis \Cs is a class of group languages (\emph{i.e.}, recognized by a finite group) with decidable separation~\cite{concagroup}. Altogether, these results generalize most of the known results regarding the decidability of levels in concatenation hierarchies.

\noindent\textbf{Contributions}. This paper is a continuation of these research efforts. Instead of looking at levels within hierarchies, we investigate the star-free closure as a whole. First, we show that the star-free closure of a finite class has decidable separation. We then use this result to establish our main theorem: the star-free closure of a class of group languages with \emph{decidable separation} has also decidable separation. In both cases, we actually prove the decidability of a stronger property called covering. Let us mention some important features of this work.

A first point is that the case of a finite class is important by itself. Foremost, it is a crucial step for the main result on the star-free closure of classes of group languages. Second, it yields a new proof that covering is decidable for the star-free languages (this is shown in~\cite{pzfo} or can be derived from~\cite{Henckell88,MR1709911}). This new proof is simpler and more generic. While the original underlying technique goes back to Wilke~\cite{wfo}, the proof has been simplified at several levels.
The main simplification is obtained thanks to an abstract framework, introduced in~\cite{pzcovering2}. It is based on the central notion of \ratm, which is meant to measure the quality of a separator. For the framework to be relevant, we actually need to generalize separation to multiple input languages, which leads to the covering problem. Another key difference is that previously existing proofs (specific to the star-free languages) involve abstracting words by new letters at some point, which requires a relabeling procedure and a change of alphabet. Here, we cannot use this approach as the classes we build with star-free closure are less robust in general. We work with a fixed alphabet, which also makes the proof simpler.

A crucial ingredient in the proof is the notion of prefix code with bounded synchronization delay. Generalizing a definition of Schützenberger~\cite{schutzbd} which was also considered by Diekert and Walter~\cite{DiekertW16a,DiekertW17}, we define a new closure operator that permits Kleene stars on such languages (this is a semantic property). This yields an operator that happens to coincide with the star-free closure when applied to the classes that we investigate. It serves as a key intermediary:  in our proofs, we heavily rely on Kleene stars to construct languages.
We therefore present this important step in the body of the paper (Theorem~\ref{thm:sfclos:carac}). Moreover, its proof provides yet another characterization of~\sfp{\Cs}, which is effective when the class~\Cs is finite (thus generalizing Schützenberger's membership result). At last regarding membership, it is worth pointing out that not only do we cover more cases, but also that it is straightforward to reprove the known algebraic characterizations from our results (see \emph{e.g.},~\cite{MIXBARRINGTON1992478}).

Finally, let us present important applications of our main result applying to input classes made of group languages. First, one may look at the input class containing all group languages. Straubing~\cite{STRAUBING1979319} described an algebraic counterpart of the star-free closure of this class, which was then shown to be recursive by Rhodes and Karnofsky~\cite{Karnofsky1982}. Altogether, this implies that membership is decidable for the star-free closure of group languages, as noted by Margolis and Pin~\cite{MargolisP85}. Here, we are able to generalize this result to separation and covering as separation is known to be decidable for the group languages~\cite{Ash91}.

Another important application is the class of languages definable by first-order logic with modular predicates \fowm. This class is known to have decidable membership~\cite{MIXBARRINGTON1992478}. Moreover, it is the star-free closure of the class consisting of the languages counting the length of words modulo some number. Since this input class is easily shown to have decidable separation (see~\cite{concagroup} for example), our main theorem applies.

The third application concerns first-order logic endowed with predicates counting the number of occurrences of a letter before a position, modulo some integer. Indeed, the class of languages definable in this logic is exactly the star-free closure of the class of languages recognized by Abelian groups (this follows from a generic correspondence theorem between star-free closure of a class and variants of first-order logic~\cite{pzgenconcat,pinbridges}, as well as from the description of languages recognized by Abelian groups~\cite{Eilenberg_book_B}). Again, our main theorem applies, since the class of Abelian groups is known to have decidable separation: this follows from~\cite{abelian_pt,MR1709911}.

\smallskip\noindent\textbf{Organization}. In Section~\ref{sec:prelims}, we recall some useful background. Section~\ref{sec:carac} presents a generic characterization of star-free closure. Then, Sections~\ref{sec:finite} and~\ref{sec:units} are devoted to our two main theorems applying respectively to finite input classes and those made of group languages. Due to space limitations, several proofs are postponed to the appendix.

\section{Preliminaries}
\label{sec:prelims}
We fix a finite alphabet $A$ for the whole paper. As usual, $A^*$ denotes the set of all words over $A$, including the empty word~\veps. For $u,v \in A^*$, we denote by $uv$ the word obtained by concatenating $u$ and~$v$. A \emph{language} is a subset of $A^*$. We lift concatenation to languages: for $K,L \subseteq A^*$, we let $KL = \{uv \mid u \in K \text{\;and\;} v \in L\}$. Finally, we use Kleene star: if $K \subseteq A^*$, $K^+$ denotes the union of all languages $K^n$ for $n \geq 1$ and $K^* = K^+ \cup \{\veps\}$.

A \emph{class of languages} is a set of languages.  A class \Cs is a \emph{Boolean algebra} when it is closed under union, intersection and complement. Moreover, \Cs is \emph{quotient-closed} if for every $L \in \Cs$ and $w \in A^*$, the languages $w^{-1}L  \stackrel{\text{def}}= \{u \in A^* \mid wu \in L\}$  and $Lw^{-1} \stackrel{\text{def}}= \{u \in A^* \mid uw \in L\}$ belong to $\Cs$. All classes considered in the paper are \varis containing only \emph{regular languages} (this will be implicit in our statements). These are the languages that can be equivalently defined by monadic second-order logic, finite automata or finite monoids. We briefly recall the monoid-based definition below.

We shall often consider \emph{finite} \varis. If \Cs is such a class, one may associate a canonical equivalence \canec over $A^*$. For $w,w'\in A^*$,  $w \canec w'$ if and only if $w \in L \Leftrightarrow w' \in L$ for every $L \in \Cs$. Moreover, we write $\ctype{w} \in \sclac$ for the \canec-class of $w$. One may then verify that the languages in \Cs are exactly the unions of \canec-classes. Moreover, since \Cs is quotient-closed, \canec is a congruence for word concatenation  (see~\cite{pzgenconcat} for proofs).

\smallskip\noindent
{\bf Regular languages.} A \emph{monoid} is a set $M$ endowed with an associative multiplication $(s,t)\mapsto s\cdot t$ (also denoted by~$st$) having a neutral element $1_M$. An \emph{idempotent} of a monoid $M$ is an element $e \in M$ such that $ee = e$. It is folklore that for any \emph{finite} monoid $M$, there exists a natural number $\omega(M)$ (denoted by $\omega$ when $M$ is understood) such that $s^\omega$ is an idempotent for every $s \in M$. Observe that $A^{*}$ is a monoid whose multiplication is concatenation (the neutral element is $\varepsilon$). Thus, we may consider monoid morphisms $\alpha: A^* \to M$ where $M$ is an arbitrary monoid. Given such a morphism and $L \subseteq A^*$, we say that $L$ is \emph{recognized} by~$\alpha$ when there exists a set $F \subseteq M$ such that $L = \alpha\inv(F)$. A language $L$ is \emph{regular} if and only if it is recognized by a morphism into a \emph{finite} monoid. Moreover, it is known that there exists a canonical recognizer of $L$, which can be computed from any representation of $L$ (such as a finite automaton): the syntactic morphism of $L$. We refer the reader to~\cite{pingoodref} for details.

\smallskip\noindent
{\bf Group languages.}  A group is a monoid $G$ in which every element $g \in G$ has an inverse $g\inv \in G$, \emph{i.e.}, \hbox{$gg\inv = g\inv g = 1_G$}. A ``\emph{group language}'' is a language $L$ recognized by a morphism into a \emph{finite group}. All classes of group languages investigated here are \varis. Typically, publications on the topic consider \emph{varieties} of group languages which is more restrictive: they involve an additional closure property called ``inverse morphic image'' (see~\cite{pinbridges}). For example, the class \md described below is \emph{not} a variety.

\begin{example}\label{ex:allgroups}
  A simple example of \vari of group languages is the class of \emph{all} group languages: \grp. Another one is \md, which contains the Boolean combinations of languages $\{w \in A^* \mid |w| = k \mod m\}$ with $k,m \in \nat$ such that $k < m$.
\end{example}

\noindent
{\bf Decision problems.} We rely on three decision problems to investigate classes of languages. Each one depends on a parameter class \Cs, which we fix for the definition. The first problem, \emph{\Cs-membership}, takes a single regular language $L$ as input and asks whether $L \in \Cs$.

The second one, \emph{\Cs-separation}, takes two regular languages $L_1$ and $L_2$ as input and asks whether $L_1$ is \Cs-separable from $L_2$ (is there a third language $K \in \Cs$ such that $L_1 \subseteq K$ and $L_2 \cap K = \emptyset$). This generalizes membership: $L \in \Cs$ if and only if $L$ is \Cs-separable from $A^* \setminus L$.

The third problem, \emph{\Cs-covering} was introduced in~\cite{pzcovering2}. Given a language $L$, a \emph{cover of $L$} is a \emph{\bf finite} set of languages \Kb such that $L \subseteq \bigcup_{K \in \Kb} K$. Moreover, a \Cs-cover of $L$  is a cover \Kb of $L$ such that all $K \in \Kb$ belong to \Cs. Consider a pair $(L_1,\Lb_2)$ where $L_1$ is a language and $\Lb_2$ is a \emph{finite set of languages}. We say that $(L_1,\Lb_2)$ is \emph{\Cs-coverable} when there exists a \Cs-cover \Kb of $L_1$ such that for every $K \in \Kb$, there exists $L \in \Lb_2$ satisfying $K \cap L = \emptyset$. The \emph{\Cs-covering problem} takes as input a single regular language $L_1$ and a finite set of regular languages $\Lb_2$. It asks whether $(L_1,\Lb_2)$ \Cs-coverable. Covering generalizes separation if \Cs is closed under union: $L_1$ is \Cs-separable from $L_2$, if and only if $(L_1,\{L_2\})$ is \Cs-coverable (see~\cite{pzcovering2}).

\medskip
\noindent
{\bf Star-free closure and main results.} We investigate an operation defined on classes: \emph{star-free closure}. Consider a class \Cs. The \emph{star-free closure of \Cs}, denoted by \sfp{\Cs}, is the least class containing \Cs and the singletons $\{a\}$ for every $a \in A$, and closed under Boolean operations and concatenation. It is standard and simple to verify that when \Cs is a \vari (which will always be the case here), this is also the case for \sfp{\Cs}.

Our main theorems state conditions on the input class \Cs guaranteeing decidability of our decision problems for \sfp{\Cs}.  First, we may handle \emph{finite} classes.

\begin{theorem}\label{thm:sfclos:main}
  Let \Cs be a finite \vari. Then, membership, separation and covering are decidable for \sfp{\Cs}.
\end{theorem}

The second theorem applies to input classes made of group languages.

\begin{theorem}\label{thm:sfclos:gmain}
  Let \Cs be a \vari of group languages with decidable separation. Then, membership, separation and covering are decidable for \sfp{\Cs}.
\end{theorem}

The remainder of the paper is devoted to proving these theorems. We first focus on \sfp{\Cs}-\emph{membership} in Section~\ref{sec:carac}. Naturally, this is weaker than directly handling \sfp{\Cs}-covering. Yet, detailing membership independently allows to introduce many proof ideas and techniques that are needed to prove the ``full'' theorems. We detail these theorems in  Sections~\ref{sec:finite} and~\ref{sec:units}. We only present the algorithms: proofs are deferred to the
appendix.

\section{Bounded synchronization delay and algebraic characterization}
\label{sec:carac}
This section is devoted to \sfp{\Cs}-membership. We handle it with a generic algebraic characterization of \sfp{\Cs} (effective under the hypotheses of Theorems~\ref{thm:sfclos:main} and~\ref{thm:sfclos:gmain}), generalizing earlier work by Pin, Straubing and Th\'erien~\cite{STRAUBING1979319,Pinambigu}. We rely on an alternate definition of star-free closure involving a semantic restriction of the Kleene star, which we first present.

\subsection{Bounded synchronization delay}

We define a second operation on classes of languages $\Cs \mapsto \bsdp{\Cs}$. We shall later prove~that it coincides with star-free closure (provided that \Cs satisfies mild hypotheses). It is based on the work of Sch\"utzenberger~\cite{schutzbd} who~defined a single class \bsd corresponding to the star-free languages (\emph{i.e.}, \sfp{\{\emptyset,A^*\}}). Here, we generalize it as an operation. The definition involves a semantic restriction of the Kleene star operation on languages: it may only be applied to ``\emph{prefix codes with bounded synchronization delay}''. Introducing this notion requires basic definitions from coding theory that we first~recall.

A language $K \subseteq A^*$ is a \emph{prefix code} when $\veps \not\in K$ and $K \cap KA^+ = \emptyset$ (no word in $K$ has a strict prefix in $K$). Note that this implies the following weaker property that we shall use implicitly: every $w \in K^*$ admits a \emph{unique} decomposition $w = w_1 \cdots w_n$ with $w_1,\dots,w_n \in K$ (this property actually defines \emph{codes} which are more general).

Given $d \geq 1$, a prefix code $K \subseteq A^+$ has \emph{synchronization delay $d$} if for every $u,v,w \in A^*$ such that $uvw \in K^+$ and $v \in K^d$, we have $uv \in K^+$. Finally, a prefix code $K \subseteq A^+$ has \emph{bounded synchronization delay} when it has synchronization delay $d$ for some $d \geq 1$.

\begin{example}
  Let $A = \{a,b\}$. Clearly, $\{ab\}$ is a prefix code with synchronization delay $1$: if $uvw \in (ab)^+$ and $v  = ab$, we have $uv \in (ab)^+$. Similarly, one may verify that $(aab)^*ab$ is a prefix code with synchronization delay $2$ (but not $1$). On the other hand, $\{aa\}$ does not have bounded synchronization delay. If $d \geq 1$, $a(aa)^da \in (aa)^*$ but  $a(aa)^d \not\in (aa)^*$.
\end{example}

We present the operation $\Cs \mapsto \bsdp{\Cs}$. The definition involves \emph{unambiguous concatenation}. Given $K,L \subseteq A^*$, their concatenation $KL$ is \emph{unambiguous} when every word $w \in KL$ admits a \emph{unique} decomposition $w = uv$ with $u \in K$ and $v \in L$. Given a class \Cs, \bsdp{\Cs} is the least class containing $\emptyset$ and $\{a\}$ for every $a \in A$, and closed under the following properties:
\begin{itemize}
\item \textbf{Intersection with \Cs:} if $K \in \bsdp{\Cs}$ and $L \in \Cs$, then $K \cap L \in \bsdp{\Cs}$.
\item \textbf{Disjoint union:} if $K,L \in \bsdp{\Cs}$ are disjoint, then $K \uplus L \in \bsdp{\Cs}$.
\item \textbf{Unambiguous product:} if $K,L \in \bsdp{\Cs}$ and $KL$ is unambiguous, then $KL \in \bsdp{\Cs}$.
\item \textbf{Kleene star for prefix codes with bounded synchronization delay:} if $K \in \bsdp{\Cs}$ is a prefix code with bounded synchronization delay, then $K^* \in \bsdp{\Cs}$.
\end{itemize}

\begin{remark}\label{rem:sfclos:bsd}
  Sch\"utzenberger proved in~\cite{schutzbd} that $\bsdp{\{\emptyset,A^*\}} = \sfp{\{\emptyset,A^*\}}$. His definition of \bsdp{\{\emptyset,A^*\}} was slightly less restrictive than ours: it does not require that the unions are disjoint and the concatenations unambiguous. It will be immediate from the correspondence with star-free closure that the two definitions are equivalent. \end{remark}

\begin{remark}
  This closure operation is different from standard ones. Instead of requiring that $\Cs \subseteq \bsdp{\Cs}$, we impose a stronger requirement: intersection with languages in \Cs is allowed. If we only asked that $\Cs \subseteq \bsdp{\Cs}$, we would get a weaker operation which does not correspond to star-free closure in general. For example, let $A = \{a,b\}$ and consider the class \md of Example~\ref{ex:allgroups}. Observe that $(aa)^* \in \bsdp{\md}$. Indeed, $\{a\} \in \bsdp{\md}$ has bounded synchronization delay, $(AA)^* \in \md$ and $(aa)^* = a^* \cap (AA)^*$. Yet, one may verify that $(aa)^*$ cannot be built from the languages of \md with union, concatenation and Kleene star applied to prefix codes with bounded synchronization delay.
\end{remark}

\subsection{Algebraic characterization of star-free closure}

We now reduce deciding membership for \sfp{\Cs} to computing \Cs-\emph{stutters}. Let us first define this new notion. Let \Cs be a \vari and $\alpha: A^* \to M$ be a morphism. A \emph{\Cs-stutter} for $\alpha$ is an element $s \in M$ such that for every \Cs-cover \Kb of $\alpha\inv(s)$, there exists $K\in\Kb$ satisfying $K\cap KK \neq \emptyset$. When $\alpha$ is understood, we simply speak of a \textit{}\Cs-stutter. Finally, we say that $\alpha$ is \emph{\Cs-aperiodic} when for every \Cs-stutter $s \in M$, we have $s^\omega = s^{\omega+1}$. The reduction is stated in the following theorem.

\begin{theorem}\label{thm:sfclos:carac}
  Let \Cs be a \vari and consider a regular language $L \subseteq A^*$. The following properties are equivalent:
\begin{enumerate}
  \item $L \in \sfp{\Cs}$.
  \item $L \in \bsdp{\Cs}$.
  \item The syntactic morphism of $L$ is \Cs-aperiodic.
  \end{enumerate}
\end{theorem}

Naturally, the characterization need not be effective: this depends on \Cs. Deciding whether a morphism is \Cs-aperiodic boils down to computing \Cs-stutters. Yet, this is possible under the hypotheses of Theorems~\ref{thm:sfclos:main} and~\ref{thm:sfclos:gmain}. First, if \Cs is a finite \vari, deciding whether an element is a \Cs-stutter is simple: there are finitely many \Cs-covers and we may check them all. If \Cs is a \vari of group languages, the question boils down to \Cs-separation as stated in the next lemma (proved in the appendix).

\begin{lemma}\label{lem:sfclos:grpstut}
  Let \Cs be a \vari of group languages and $\alpha: A^* \to M$~be a morphism. For all $s \in M$, $s$ is a \Cs-stutter if and only if $\{\veps\}$ is \textbf{not} \Cs-separable from~$\alpha\inv(s)$.
\end{lemma}

Altogether, we obtain the membership part in Theorems~\ref{thm:sfclos:carfinite} and~\ref{thm:sfclos:cargroup}. We conclude the section with an extended proof sketch for the most interesting direction in Theorem~\ref{thm:sfclos:carac}: $3) \Rightarrow 2)$ (a detailed proof for the two other directions is provided in appendix).

\begin{proof}[Proof of $3) \Rightarrow 2)$ in Theorem~\ref{thm:sfclos:carac}] Let \Cs be a \vari and \mbox{$\alpha: A^* \to M$} be a \Cs-aperiodic morphism. We show that all languages recognized by $\alpha$ belong to \bsdp{\Cs}.

Given $K \subseteq A^*$ and $s\in M$, we say that $K$ is \emph{$s$-safe} when $s\alpha(u)=s\alpha(v)$ for every $u,v \in K$. We extend this notion to sets of languages: such a set \Kb is $s$-safe when every $K \in \Kb$ is $s$-safe. We shall use $s$ as an induction parameter. Finally, given a language $P \subseteq A^*$, an \bsdp{\Cs}-partition of $P$ is a finite partition of $P$ into languages of \bsdp{\Cs}. 

  \begin{proposition}\label{prop:sfclos:fromapertostar}
    Let $P \subseteq A^+$ be a prefix code with bounded synchronization delay. Assume that there exists a $1_M$-safe \bsdp{\Cs}-partition of $P$. Then, for every $s \in M$, there exists an $s$-safe \bsdp{\Cs}-partition of $P^*$.
  \end{proposition}

  We first apply Proposition~\ref{prop:sfclos:fromapertostar} to conclude the main argument. We show that every language recognized by $\alpha$ belongs to \bsdp{\Cs}. By definition, \bsdp{\Cs} is closed under disjoint union. Hence, it suffices to show that $\alpha\inv(t) \in \bsdp{\Cs}$ for every $t \in M$. We fix $t \in M$ for the proof.

  Clearly, $A \subseteq A^+$ is a prefix code with bounded synchronization delay and  $\{\{a\} \mid a \in A\}$ is a $1_M$-safe \bsdp{\Cs}-partition of $A$. Hence, Proposition~\ref{prop:sfclos:fromapertostar} (applied in the case $s= 1_M$) yields a $1_M$-safe \bsdp{\Cs}-partition \Kb of $A^*$. One may verify that $\alpha\inv(t)$ is the disjoint union of all $K \in \Kb$ intersecting $\alpha\inv(t)$. Hence, $\alpha\inv(t) \in \bsdp{\Cs}$ which concludes the main argument.

  \medskip

  It remains to prove Proposition~\ref{prop:sfclos:fromapertostar}. We let $P \subseteq A^*$ be a prefix code with bounded synchronization delay, \Hb a $1_M$-safe \bsdp{\Cs}-partition of $P$ and $s \in M$. We need to build an \bsdp{\Cs}-partition \Kb of $P^*$ such that every $K \in \Kb$ is $s$-safe. We proceed by induction on the three following parameters listed by order of importance: $(1)$ the size of $\alpha(P^+) \subseteq M$, $(2)$ the size of $\Hb$ and $(3)$ the size of $s \cdot \alpha(P^*) \subseteq M$.  We distinguish two cases depending on the following property of $s$ and $\Hb$. We say that \emph{$s$ is \Hb-stable} when the following holds:
  \begin{equation}\label{eq:sfclos:mostable}
    \text{for every $H \in \Hb$,} \quad  s \cdot \alpha(P^*) = s \cdot \alpha(P^*H).
  \end{equation}
  The base case happens when $s$ is \Hb-stable. Otherwise, we use induction on our parameters.

  \smallskip
  \noindent
  \textbf{Base case: $s$ is \Hb-stable.} Since $\alpha$ is \Cs-aperiodic, we have the following simple fact.

  \begin{fact}\label{fct:caracfin}
    There is a \emph{finite} \vari $\Ds \subseteq \Cs$ such that $\alpha$ is \Ds-aperiodic.
  \end{fact}

  Since \Ds is finite, we may consider the associated canonical equivalence \caned over $A^*$. We let $\Kb = \{P^* \cap D \mid D \in \dclac\}$. Clearly, \Kb is a partition of $P^*$. Let us verify that it only contains languages in \bsdp{\Cs}. We have $P \in \bsdp{\Cs}$: it is the disjoint union of all languages in the \bsdp{\Cs}-partition \Hb of $P$. Moreover, $P^* \in \bsdp{\Cs}$ since $P$ is a prefix code with bounded synchronization delay. Hence, $P^* \cap D \in \bsdp{\Cs}$ for every $D \in \dclac$ since $D \in \Ds \subseteq \Cs$. Therefore, it remains to show that every language $K \in \Kb$ is $s$-safe. This is a consequence of the following lemma which is proved using the hypothesis~\eqref{eq:sfclos:mostable} that $s$ is \Hb-stable.

  \begin{lemma}\label{lem:sfclos:godcase1}
    For every $u,v \in P^*$ such that $u \caned v$, we have $s \alpha(u) = s \alpha(v)$.
  \end{lemma}

  \smallskip
  \noindent
  \textbf{Inductive step: $s$ is not \Hb-stable.} By hypothesis, we know that~\eqref{eq:sfclos:mostable} does not hold. Therefore, we get some $H \in \Hb$ such that the following \textbf{strict} inclusion holds,
  \begin{equation}\label{eq:sfclos:godinduc}
    s \cdot \alpha(P^*H) \subsetneq s \cdot \alpha(P^*).
  \end{equation}
  We fix this language $H \in \Hb$ for the remainder of the proof. The following lemma is proved by induction on our second parameter (the size of $\Hb$).

  \begin{lemma}\label{lem:sfclos:alphind}
    There exists a $1_M$-safe \bsdp{\Cs}-partition \Ub of $(P \setminus H)^*$.
  \end{lemma}

  We fix the partition \Ub of $(P \setminus H)^*$ given by Lemma~\ref{lem:sfclos:alphind} and distinguish two independent sub-cases. Since $H \subseteq P$ (as $H$ is an element of the partition \Hb of $P$), we have $\alpha(P^*H) \subseteq \alpha(P^+)$. We use a different argument depending on whether this inclusion is strict or not.

  \smallskip
  \noindent
  \textbf{Sub-case~1: $\alpha(P^*H) = \alpha(P^+)$.} Since $H$ is $1_M$-safe by hypothesis, there exists $t \in M$~such that $\alpha(H) = \{t\}$. Similarly, since every $U \in \Ub$ is $1_M$-safe, there exists $r_U \in M$ such that $\alpha(U) = \{r_U\}$. The construction of \Kb is based on the next lemma which is proved using~\eqref{eq:sfclos:godinduc}, the hypothesis of Sub-case~1 and induction on our third parameter (the size of $s \cdot \alpha(P^*) \subseteq M$).

  \begin{lemma}\label{lem:sfclos:sc1carac}
    For every $U \in \Ub$, there exists an $sr_Ut$-safe \bsdp{\Cs}-partition $\Wb_{U}$ of $P^*$.
  \end{lemma}

  We are ready to define the partition \Kb of $P^*$. Using Lemma~\ref{lem:sfclos:sc1carac}, we define,
  \[
    \Kb = \Ub \cup \bigcup_{U \in \Ub} \{UHW \mid W \in \Wb_U\}
  \]
  It remains to show that \Kb is an $s$-safe \bsdp{\Cs}-partition of~$P^*$. First, \Kb is a partition of~$P^*$ since $P$ is a prefix code and $H \subseteq P$. Indeed, every word $w \in P^*$ admits a \emph{unique} decomposition $w = w_1 \cdots w_n$ with $w_1,\dots,w_n \in P$. If no factor $w_i$ belongs to $H$, then $w \in (P \setminus H)^*$ and $w$ belongs to some unique $U \in \Ub$. Otherwise, let $w_i$ be the leftmost factor such that $w_i \in H$. Thus, $w_1 \cdots w_{i-1} \in (P \setminus H)^*$, which also yields a unique $U \in \Ub$ such that $w_1 \cdots w_{i-1} \in U$ and $w_{i+1} \cdots w_n \in P^*$ which yields a unique $W \in \Wb_U$ such that $w_{i+1} \cdots w_n \in W$. Thus, $w \in UHW$ which is an element of \Kb (the only one containing $w$).

  Moreover, every $K \in \Kb$ belongs to \bsdp{\Cs}. If $K \in \Ub$, this is immediate by definition of \Ub in Lemma~\ref{lem:sfclos:alphind}. Otherwise, $K = UHW$ with $U \in \Ub$ and $W \in \Wb_U$. We know that $U,H,W \in \bsdp{\Cs}$ by definition. Moreover, one may verify that the concatenation $UHW$ is \emph{unambiguous} since $P$ is a prefix code, $U \subseteq (P \setminus H)^*$ and $W \subseteq H^*$. Hence, $K \in \bsdp{\Cs}$.

  Finally, we verify that \Kb is $s$-safe. Consider $K \in \Kb$ and $w,w' \in K$, we show that $s\alpha(w) = s\alpha(w')$. If $K \in \Ub$, this is immediate: \Ub is $1_M$-safe by definition. Otherwise, $K=UHW$ with $U \in \Ub$ and $W \in \Wb_U$. By definition, $\alpha(H) = \{t\}$ and $\alpha(U) = \{r_U\}$ which implies that $s\alpha(w) = sr_U t\alpha(x)$ and $s\alpha(w') = sr_U t\alpha(x')$ for $x,x' \in W$. Moreover, $W \in \Wb_U$ is $sr_Ut$-safe by definition. Hence, $s\alpha(w) = s \alpha(w')$, which concludes the proof of this sub-case.

  \smallskip
  \noindent
  \textbf{Sub-case~2: $\alpha(P^*H) \subsetneq \alpha(P^+)$.} Consider $w \in P^*$. Since $P$ is a prefix code, $w$ admits a unique decomposition $w = w_1 \cdots w_n$ with $w_1,\dots,w_n \in P$. We may look at the rightmost factor $w_i \in H \subseteq P$ to uniquely decompose $w$ in two parts (each of them possibly empty): the prefix $w_1 \cdots w_i \in ((P \setminus H)^*H)^*$ and the suffix in $w_{i+1} \cdots w_n \in (P \setminus H)^*$. Using induction,~we construct \bsdp{\Cs}-partitions of the possible languages of prefixes and suffixes. Then, we combine them to construct a partition of the whole set $P^*$. We already handled the suffixes: \Ub is an \bsdp{\Cs}-partition of $(P \setminus H)^*$. The prefixes are handled using the hypothesis of Sub-case~2 and induction on our first parameter (the size of~$\alpha(P^+)$).

  \begin{lemma}\label{lem:sfclos:sc2carac}
    There exists a $1_M$-safe \bsdp{\Cs}-partition \Vb of $((P \setminus H)^*H)^*$.
  \end{lemma}

  Using Lemma~\ref{lem:sfclos:sc2carac}, we define $\Kb = \{VU \mid V \in \Vb \text{ and } U \in \Ub\}$. It follows from the above discussion that \Kb is a partition of $P^*$ since  \Vb and \Ub are partitions of $((P \setminus H)^*H)^*$ and $(P \setminus H)^*$, respectively. Moreover, every $K \in \Kb$ belongs to \bsdp{\Cs}: $K = VU$ with $V \in \Vb$ and $U\in \Ub$, and one may verify that this is an \emph{unambiguous} concatenation. It remains to show that \Kb is $s$-safe. Let $K \in \Kb$ and $w,w' \in K$. We show that $s\alpha(w) = s\alpha(w')$. By definition, we have $K = VU$ with $V \in \Vb$ and $U \in \Ub$. Therefore, $w =vu$ and $w' = v'u'$ with $u,u' \in U$ and $v,v' \in V$. Since $U$ and $V$ are both $1_M$-safe by definition, we have $\alpha(u) = \alpha(u')$ and $\alpha(v) = \alpha(v')$. It follows that $s\alpha(w) = s\alpha(w')$, which concludes the proof of Proposition~\ref{prop:sfclos:fromapertostar}.
\end{proof}

\section{Covering when the input class is finite}
\label{sec:finite}
This section is devoted to Theorem~\ref{thm:sfclos:main}. We show that when \Cs is a finite \vari, \sfp{\Cs}-covering is decidable by presenting a generic algorithm. It is formulated within a framework designed to handle covering questions, which was originally introduced in~\cite{pzcovering2}. We start by briefly recalling it (we refer the reader to~\cite{pzcovering2} for details).

\subsection{\Ratms and optimal \imprints}

The framework is based on an algebraic object called ``\ratm''. These are morphisms of commutative and idempotent monoids. We write such monoids $(R,+)$: the binary operation  ``$+$'' is called \emph{addition} and the neutral element is denoted by $0_R$. Being idempotent means that $r + r = r$ for every $r \in R$. For every commutative and idempotent monoid $(R,+)$, one may define a canonical ordering $\leq$ over $R$: for $r, s\in R$, we have $r \leq s$ when $r+s=s$. One may verify that $\leq$ is a partial order which is compatible with addition. 

\begin{example}
  For every set $E$, $(2^E,\cup)$ is an idempotent and commutative monoid. The neutral element is $\emptyset$ and the canonical ordering is inclusion.
\end{example}

A \ratm is a morphism $\rho: (2^{A^*},\cup) \to (R,+)$ where $(R,+)$ is a \emph{finite} idempotent and commutative monoid, called the \emph{rating set of $\rho$}. That is, $\rho$ is a map from $2^{A^{*}}$ to $R$ such that $\rho(\emptyset) = 0_R$ and $\rho(K_1\cup K_2)=\rho(K_1)+\rho(K_2)$ for every $K_1,K_2 \subseteq A^*$.

For the sake of improved readability, when applying a \ratm $\rho$ to a singleton set $\{w\}$, we write $\rho(w)$ for $\rho(\{w\})$. Moreover, we write $\rho_*: A^* \to R$ for the restriction of $\rho$ to $A^*$: for every $w \in A^*$, we have $\rho_*(w) = \rho(w)$ (this notation is useful when referring to the language $\rho_*\inv(r) \subseteq A^*$, which consists of all words $w \in A^*$ such that $\rho(w) = r$).

\smallskip

Most of the theory makes sense for arbitrary \ratms. However, we shall often have to work with special \ratms satisfying additional properties. We define two kinds.

\smallskip
\noindent
{\bf \Nice \ratms.} A \ratm $\rho: 2^{A^*} \to R$ is \emph{\nice} when, for every nonempty language $K \subseteq A^*$, there exist finitely many words $w_1,\dots,w_n \in K$ such that $\rho(K) = \rho(w_1) + \cdots + \rho(w_k)$.

When a \ratm $\rho: 2^{A^*} \to R$ is \nice, it is characterized by the canonical map $\rho_*: A^* \to R$. Indeed, for $K \subseteq A^*$, we may consider the sum of all elements $\rho(w)$ for $w \in K$: while it may be infinite, this sum boils down to a finite one since $R$ is commutative and idempotent. The hypothesis that $\rho$ is \nice implies that $\rho(K)$ is equal to this sum.

\smallskip
\noindent
{\bf \Mratms.} A \ratm $\rho: 2^{A^*} \to R$ is \tame when its rating set $R$ has more structure: it needs to be an \emph{idempotent semiring}. A \emph{semiring} is a tuple $(R,+,\cdot)$ where $R$ is a set and ``$+$'' and ``$\cdot$''  are two binary operations called addition and multiplication. Moreover, $(R,+)$ is a commutative monoid, $(R,\cdot)$ is a monoid (the neutral element is denoted by $1_R$), the multiplication distributes over addition and the neutral element ``$0_R$'' of $(R,+)$ is a zero for $(R,\cdot)$ ($0_R \cdot r = r \cdot 0_R = 0_R$ for every $r \in R$). A semiring $R$ is \emph{idempotent} when $r + r = r$ for every $r \in R$, \emph{i.e.}, when the additive monoid $(R,+)$ is idempotent (there is no additional constraint on the multiplicative monoid $(R,\cdot)$).

\begin{example}
  A key example of an infinite idempotent semiring is the set $2^{A^*}$. Union is the addition and language concatenation is the multiplication (with $\{\varepsilon\}$ as neutral element).
\end{example}

Let $\rho: 2^{A^*} \to R$ be a \ratm: $(R,+)$ is an idempotent commutative monoid and $\rho$ is a morphism from $(2^{A^*},\cup)$ to $(R,+)$. We say that $\rho$ is \emph{\tame} when the rating set $R$ is equipped with a multiplication ``$\cdot$'' such that $(R,+,\cdot)$ is an idempotent semiring and $\rho$ is also a monoid morphism from $(2^{A^*},\cdot)$ to $(R,\cdot)$. That is, the two following additional axioms have to be satisfied: $\rho(\varepsilon) = 1_R$ and $\rho(K_1K_2) = \rho(K_1) \cdot \rho(K_2)$ for every $K_1,K_2 \subseteq A^*$. 

\begin{remark}
  \Ratms which are both \nice and \tame are finitely representable. As we explained, if $\rho: 2^{A^*} \to R$ is \nice, it is characterized by the canonical map $\rho_*: A^* \to R$.  When $\rho$ is also \tame, $\rho_*$ is finitely representable: it is a morphism into a finite monoid. Hence, we may speak of algorithms whose input is a \nice \mratm.

  \Ratms which are not \nice and \tame cannot be finitely represented in general. Yet, they are crucial: while our main statements consider \nice \mratms, many proofs involve auxiliary \ratms which are neither \nice nor \tame.
\end{remark}

\noindent
{\bf Optimal \imprints.} Now that we have \ratms, we turn to \imprints. Consider a \ratm $\rho: 2^{A^*} \to R$. Given any finite set of languages~\Kb, we define the $\rho$-\imprint of~\Kb. Intuitively, when \Kb is a cover of some language $L$, this object measures the ``quality'' of \Kb. The $\rho$-\imprint \emph{of \Kb} is the following subset of~$R$:
\[
  \prin{\rho}{\Kb} = \{r \mid r \leq \rho(K) \text{ for some } K \in\Kb\}.
\]
We may now define optimality. Consider an arbitrary \ratm $\rho: 2^{A^*} \to R$ and a Boolean algebra \Cs. Given a language $L$, an optimal \Cs-cover of $L$ for $\rho$ is a \Cs-cover \Kb of $L$ which satisfies the following property:
\[
  \prin{\rho}{\Kb} \subseteq \prin{\rho}{\Kb'} \quad \text{for every \Cs-cover $\Kb'$ of $L$}.
\]
In general, there can be infinitely many optimal \Cs-covers for a given \ratm $\rho$. It is shown in~\cite{pzcovering2} that there always exists at least one (using closure under intersection for \Cs).

Clearly, for a Boolean algebra \Cs, a language $L$ and a \ratm $\rho$, all optimal \Cs-covers of $L$ for $\rho$ have the same $\rho$-\imprint. Hence, this unique $\rho$-\imprint is a \emph{canonical} object for \Cs, $L$ and $\rho$. We call it the \emph{\Cs-optimal $\rho$-\imprint on $L$} and we write it $\opti{\Cs}{L,\rho}$:
\[
  \opti{\Cs}{L,\rho} = \prin{\rho}{\Kb} \quad \text{for any optimal \Cs-cover \Kb of $L$  for $\rho$}.
\]
We complete the definition with a simple useful fact (a proof is available in~\cite{pzbpolc}).

\begin{fact} \label{fct:lunion}
  Let \Cs be a Boolean algebra, $\rho: 2^{A^*} \to R$ a \ratm and $L_1,L_2 \subseteq A^*$. Then, $\opti{\Cs}{L_1,\rho}\cup\opti{\Cs}{L_2,\rho}=\opti{\Cs}{L_1\cup L_2,\rho}$.
\end{fact}

\noindent
{\bf Connection with covering.} Consider the special case when the language $L$ that needs to be covered is $A^*$. In that case, we write \copti{\rho} for \copti{A^*,\rho}. It is shown in~\cite{pzcovering2} that for every \emph{Boolean algebra} \Cs, deciding \Cs-covering formally reduces to computing \Cs-optimal \imprints from input \nice \mratms.

\begin{proposition} \label{prop:breduc}
  Let \Cs be a Boolean algebra. Assume that there exists an algorithm which computes \copti{\rho} from an input \nice \mratm $\rho$. Then, \Cs-covering is decidable.
\end{proposition}

\subsection{Algorithm}

We may now present our algorithm for \sfp{\Cs}-covering when \Cs is a finite \vari. We fix \Cs for the presentation. In view of Proposition~\ref{prop:breduc}, we need to prove that one may compute \sfcopti from an input \nice \mratm $\rho$.

Our algorithm actually computes slightly more information. Since \Cs is a finite \vari, we may consider the equivalence \canec over $A^*$. In particular, the set \sclac of \canec-classes is a finite monoid (we write ``\cmult'' for its multiplication) and the map $w \mapsto \ctype{w}$ is a morphism. Given a \ratm $\rho: 2^{A^*} \to R$ we define:
\[
  \csfcopti = \{(C,r) \in (\sclac) \times R \mid r \in \opti{\sfp{\Cs}}{C,\rho}\}
\]
Observe that \csfcopti captures more information than \sfcopti. Indeed, it encodes all sets \opti{\sfp{\Cs}}{C,\rho} for $C \in \sclac$ and by Fact~\ref{fct:lunion}, \sfcopti is the union of all these sets.

Our main result is a least fixpoint procedure for computing \csfcopti from a \nice \mratm $\rho$. It is based on a generic characterization theorem which we first present. Given an arbitrary \nice \mratm $\rho: 2^{A^*} \to R$ and a set $S \subseteq (\sclac) \times R$, we say that $S$ is \emph{\sfp{\Cs}-saturated for $\rho$} when the following properties are satisfied:
\begin{enumerate}
\item {\bf Trivial elements.} For every $w \in A^*$, we have $(\ctype{w},\rho(w)) \in S$.
\item {\bf Downset.} For every $(C,r) \in S$ and $q \in R$, if $q \leq r$, then $(C,q) \in S$.
\item {\bf Multiplication.} For every $(C,q),(D,r) \in S$, we have $(C \cmult D,qr) \in S$.
\item {\bf \sfp{\Cs}-closure.} For all $(E,r) \in S$, if $E \in \sclac$ is idempotent, then $(E,r^\omega + r^{\omega+1}) \in S$.
\end{enumerate}

\begin{theorem}[\sfp{\Cs}-optimal \imprints (\Cs finite)] \label{thm:sfclos:carfinite}
  Let $\rho: 2^{A^*} \to R$ be a \emph{\nice} \mratm. Then, \csfcopti is the least \sfp{\Cs}-saturated subset of $(\sclac) \times R$ for $\rho$.
\end{theorem}

Given a \nice \mratm $\rho: 2^{A^*} \to R$ as input, it is clear that one may compute the least \sfp{\Cs}-saturated subset of $(\sclac) \times R$ with a least fixpoint procedure. Hence, Theorem~\ref{thm:sfclos:carfinite} provides an algorithm for computing \csfcopti. As we explained above, we may then compute \sfcopti from this set. Together with Proposition~\ref{prop:breduc}, this yields Theorem~\ref{thm:sfclos:main} as a corollary: \sfp{\Cs}-covering is decidable when \Cs is a finite \vari. Theorem~\ref{thm:sfclos:carfinite} is proved in the appendix.

\section{Covering when the input class is made of group languages}
\label{sec:units}
This section is devoted to Theorem~\ref{thm:sfclos:gmain}. We show that when \Cs is a \vari of group languages with decidable separation, \sfp{\Cs}-covering is decidable.

As in Section~\ref{sec:finite}, we rely on Proposition~\ref{prop:breduc}: we present an algorithm computing \sfcopti from an input \nice \mratm $\rho$. We do not work with \sfcopti itself but with another set carrying  more information. Its definition requires introducing a few additional concepts. We first present them and then turn to the algorithm. For more details, see~\cite{concagroup}.

\subsection{Preliminary definitions}

\noindent
\textbf{Optimal \idens.} In this case, handling \sfp{\Cs} involves considering \Cs-optimal covers of $\{\veps\}$. Since $\{\veps\}$ is a singleton, there always exists such a cover consisting of a single language, which leads to the following definition.

Let \Cs be a Boolean algebra (we shall use the case when \Cs contains only group languages but this is not required for the definitions) and $\tau: 2^{A^*} \to Q$ a \ratm. A \emph{\Cs-optimal \iden for $\tau$} is a language $L \in \Cs$ such that $\veps \in L$ and  $\tau(L) \leq \tau(L')$ for every $L' \in \Cs$ satisfying $\veps \in L'$. As expected, there always exists a \Cs-optimal \iden for any \ratm $\tau$ (see the appendix for a proof).

By definition, all \Cs-optimal \idens for $\tau$ have the same image under $\tau$. We write it $\ioptic{\tau} \in Q$: $\ioptic{\tau} = \tau(L)$ for every \Cs-optimal \iden $L$ for $\tau$. It turns out that when $\tau$ is \nice and \tame, computing \iopti{\Cs}{\tau} from $\tau$ boils down to \Cs-separation. This is important: this is exactly how our algorithm for \sfp{\Cs}-covering depends on \Cs-separation.

\begin{lemma}\label{lem:sepepswit}
  Let $\tau: 2^{A^*} \to Q$ be a \nice \ratm and \Cs a Boolean algebra. Then, $\iopti{\Cs}{\tau}$ is the sum of all $q \in Q$ such that $\{\veps\}$ is not \Cs-separable from $\tau_*\inv(q)$.
\end{lemma}

\noindent
\textbf{Nested \ratms.} We want an algorithm which computes \sfcopti from an input \nice \mratm $\rho$ for a fixed \vari of group languages \Cs. Yet, we shall not use optimal \idens with this input \ratm $\rho$. Instead, we consider an auxiliary \ratm built from $\rho$ (the definition is taken from~\cite{pzbpolc}).

Consider a Boolean algebra \Ds (we shall use the case $\Ds = \sfp{\Cs}$) and a \ratm $\rho: 2 ^{A^*} \to R$. We build a new map $\bratauxd: 2^{A^*} \to 2^R$ whose rating set is $(2^R,\cup)$. For every $K \subseteq A^*$, we define $\bratauxd(K) = \opti{\Ds}{K,\rho} \in 2^R$. It follows from Fact~\ref{fct:lunion} that this is indeed a \ratm (on the other hand \bratauxd need not be \nice nor \tame, see~\cite{pzbpolc} for details).

We may now explain which set is computed by our algorithm instead of \sfcopti. Consider a \nice \mratm $\rho: 2^{A^*} \to R$. Since $\bratauxsfc: 2^{A^*} \to 2^R$ is a \ratm, we may consider the element $\ioptic{\bratauxsfc} \in 2^R$. By definition, $\ioptic{\bratauxsfc} = \bratauxsfc(L)$ where $L$ is a \Cs-optimal \iden for $\bratauxsfc$. Therefore, \ioptic{\bratauxsfc} is a subset of $\bratauxsfc(A^*) = \opti{\sfp{\Cs}}{A^*,\rho} = \sfcopti$. When \Cs is a \vari of group languages, one may compute the whole set \sfcopti from this subset.

\begin{proposition}\label{prop:utooptibool}
  Let \Cs be a \vari of group languages and $\rho: 2^{A^*} \to R$ a \nice \mratm. Then, \sfcopti is the least subset of $R$ containing \ioptic{\bratauxsfc} and satisfying the three following properties:
  \begin{itemize}
  	\item {\bf Trivial elements.} For every $w \in A$, $\rho(w) \in \sfcopti$.
  	\item {\bf Downset.} For every $r \in \sfcopti$ and $q \leq r$, we have $q \in \sfcopti$.
  	\item {\bf Multiplication.} For every $q,r \in \sfcopti$, we have $qr \in \sfcopti$.
  \end{itemize}
\end{proposition}

\begin{remark}
	Intuitively, we use \ioptic{\bratauxsfc} to ``nest'' two optimizations: one for \Cs and the other for \sfp{\Cs}. Indeed, $\ioptic{\bratauxsfc} = \bratauxsfc(L) = \opti{\sfp{\Cs}}{L,\rho}$ where $L$ is a \Cs-optimal \iden for $\bratauxsfc$.  Hence, \ioptic{\bratauxsfc} is least set $\prin{\rho}{\Kb} \subseteq R$ (with respect to inclusion), over all \sfp{\Cs}-covers \Kb of some language $L \in \Cs$ containing $\veps$.
\end{remark}

\subsection{Algorithm}

\newcommand{\sfrat}[1]{\ensuremath{\eta_{\rho,#1}}\xspace}
\newcommand{\sfrats}{\sfrat{S}}

We may now present our algorithm for computing \sfcopti. We fix a \vari of group languages \Cs for the presentation. As expected, the main procedure computes \ioptic{\bratauxsfc} (see Proposition~\ref{prop:utooptibool}). In this case as well, this procedure is obtained from a characterization theorem.

Consider a \nice \mratm $\rho: 2^{A^*} \to R$. We define the \sfp{\Cs}-complete subsets of $R$ for $\rho$. The definition depends on auxiliary \nice \mratms. We first present them. Clearly, $2^R$ is an idempotent semiring (addition is union and the multiplication is lifted from the one of $R$). For every $S \subseteq R$, we use it as the rating set of a \nice \mratm $\sfrats: 2^{A^*} \to 2^{R}$. Since we are defining a \emph{\nice} \mratm, it suffices to specify the evaluation of letters. For $a \in A$, we let $\sfrats(a) = S \cdot \{\rho(a)\} \cdot S \in 2^R$. Observe that by definition, we have $\ioptic{\sfrats} \subseteq R$.

We are ready to define the \sfp{\Cs}-complete subsets of $R$. Consider $S \subseteq R$. We say that $S$ is \emph{\sfp{\Cs}-complete for $\rho$} when the following
conditions are satisfied:
\begin{enumerate}
	\item {\bf Downset.} For every $r \in S$ and $q \leq r$, we have $q \in S$.
	\item {\bf Multiplication.} For every $q,r \in S$, we have $qr \in S$.
	\item {\bf \Cs-operation.} We have $\ioptic{\sfrats} \subseteq S$.
	\item {\bf \sfp{\Cs}-closure.} For every $r \in S$, we have $r^\omega + r^{\omega+1} \in S$.
\end{enumerate}

\begin{remark} \label{rem:sfclos:compne}
	The definition of \sfp{\Cs}-complete subsets does not explicitly require that they contain some trivial elements. Yet, this is implied by \Cs-operation. Indeed, if $S \subseteq R$ is \sfp{\Cs}-complete, then $\sfrats(\veps) = \{1_R\}$ (this is the multiplicative neutral element of $2^{R}$). This implies that $1_R \in \ioptic{\sfrats}$ and we obtain from \Cs-operation that $1_R \in S$.
\end{remark}

\begin{theorem}[\sfp{\Cs}-optimal \imprints (\Cs made of group languages)] \label{thm:sfclos:cargroup}
	Let $\rho: 2^{A^*} \to R$ be a \nice \mratm. Then, \ioptic{\bratauxsfc} is the least \sfp{\Cs}-complete subset of $R$.	
\end{theorem}

When \Cs-separation is decidable, Theorem~\ref{thm:sfclos:cargroup} yields a least fixpoint procedure for computing \ioptic{\bratauxsfc} from a \nice \mratm $\rho: 2^{A^*} \to R$. The computation starts from the empty set and saturates it with the four operations in the definition of \sfp{\Cs}-complete subsets. It is clear that we may implement downset, multiplication and \sfp{\Cs}-closure. Moreover, we may implement \Cs-operation as this boils down to \Cs-separation by Lemma~\ref{lem:sepepswit}. Eventually, the computation reaches a fixpoint and it is straightforward to verify that this set is the least \sfp{\Cs}-complete subset of $R$, \emph{i.e.}, \ioptic{\bratauxsfc} by Theorem~\ref{thm:sfclos:cargroup}.

\smallskip

By Proposition~\ref{prop:utooptibool}, we may compute \sfcopti from \ioptic{\bratauxsfc}. Altogether, this yields the decidability of \sfp{\Cs}-covering by Proposition~\ref{prop:breduc}. Hence, Theorem~\ref{thm:sfclos:gmain} is proved.

\section{Conclusion}
We proved that for any \vari \Cs, \sfp{\Cs}-covering is decidable whenever \Cs is either finite or made of group languages and with decidable separation. Moreover, we presented an algebraic characterization of \sfp{\Cs} which holds for every \vari \Cs, generalizing earlier results~\cite{STRAUBING1979319,Pinambigu}. A key proof ingredient is an alternative definition of star-free closure: the operation $\Cs \mapsto \bsdp{\Cs}$ which we prove to be equivalent. This correspondence generalizes the work of Sch\"utzenberger~\cite{schutzbd} who introduced a single class \bsd ({\em i.e.} \bsdp{\{\emptyset,A^*\}}) corresponding to the star-free languages ({\em i.e.}~\sfp{\{\emptyset,A^*\}}).

Our results can be instantiated for several input classes \Cs. Theorem~\ref{thm:sfclos:main} applies when~\Cs is finite. In this case, the only prominent application is the class of star-free languages itself. It was already known that covering is decidable for this class~\cite{Henckell88,pzfo}. However, Theorem~\ref{thm:sfclos:main} is important for two reasons. First, its proof is actually simpler than the earlier ones specific to the star-free languages (this is achieved by relying on the operation $\Cs \mapsto \bsdp{\Cs}$). More importantly, Theorem~\ref{thm:sfclos:main} is used as a key ingredient for proving our second generic statement: Theorem~\ref{thm:sfclos:gmain}, which applies to classes made of group languages with decidable separation. It is known that separation is decidable for the class \grp of all group languages~\cite{Ash91}. Hence, we obtain that \sfp{\grp}-covering is decidable. Another application is the class \md consisting of languages counting the length of words modulo some number (deciding \md-separation is a simple exercise). We get the decidability of \sfp{\md}-covering. This is important, as the languages in \sfp{\md} are those definable in first-order logic with modular predicates (\fowm). A last example is given by the input class consisting of all languages counting the number of occurrences of letters modulo some number. These are exactly the languages recognized by finite commutative groups, for which separation is decidable~\cite{abelian_pt}.

\newpage
\appendix

\section{Appendix to Section~\ref{sec:carac}}
\label{app:sfc}
This appendix presents the missing proofs in Section~\ref{sec:carac}. It is written as a self-contained, full version of Section~\ref{sec:carac} which can be read independently.

We handle \sfp{\Cs}-membership with a generic algebraic characterization of \sfp{\Cs} (effective under the hypotheses of Theorems~\ref{thm:sfclos:main} and~\ref{thm:sfclos:gmain}), generalizing earlier work by Pin, Straubing and Th\'erien~\cite{STRAUBING1979319,Pinambigu}. We rely on an alternate definition of star-free closure involving a semantic restriction of the Kleene star, which we first present.

\subsection{Bounded synchronization delay}

We define a second operation on classes of languages $\Cs \mapsto \bsdp{\Cs}$. We shall later prove that it coincides with star-free closure. It is based on the work of Sch\"utzenberger~\cite{schutzbd} who defined a single class \bsd corresponding to the star-free languages (\emph{i.e.}, \sfp{\{\emptyset,A^*\}}). Here, we generalize it as an operation. The definition involves a semantic restriction of the Kleene star operation on languages: it may only be applied to ``\emph{prefix codes with bounded synchronization delay}''. Introducing this notion requires basic definitions from code theory that we first recall.

A language $K \subseteq A^*$ is a \emph{prefix code} when $\veps \not\in K$ and $K \cap KA^+ = \emptyset$ (no word in $K$ has a strict prefix in $K$). 
Given $d \geq 1$, a prefix code $K \subseteq A^+$ has \emph{synchronization delay $d$} if for every $u,v,w \in A^*$ such that $uvw \in K^+$ and $v \in K^d$, we have $uv \in K^+$. Finally, a prefix code $K \subseteq A^+$ has \emph{bounded synchronization delay} when it has synchronization delay $d$ for some $d \geq 1$.

\begin{example}
	Let $A = \{a,b\}$. Clearly, $\{ab\}$ is a prefix code with synchronization delay $1$: if $uvw \in (ab)^+$ and $v  = ab$, we have $uv \in (ab)^+$. Similarly, one may verify that $(aab)^*ab$ is a prefix code with synchronization delay $2$ (but not $1$). On the other hand, $\{aa\}$ does not have bounded synchronization delay. If $d \geq 1$, $a(aa)^da \in (aa)^*$ but  $a(aa)^d \not\in (aa)^*$.
\end{example}

We complement the definition with a few standard properties of prefix codes with bounded synchronization delay. They will be useful when manipulating them later. First, we have the following fact about prefix codes.

\begin{fact} \label{fct:sfclos:prefdec}
	Let $K$ be a prefix code. Consider $m,n \in \nat$, $u_1,\dots,u_m \in K$ and $v_1 \cdots,v_n \in K$. The two following properties hold:
	\begin{itemize}
		\item If $u_1 \cdots u_m$ is a prefix of $v_1 \cdots v_n$, then $m \leq n$ and $u_i = v_i$ for every $i \leq m$.
		\item If $u_1 \cdots u_m = v_1 \cdots v_n$, then $m = n$ and $u_i = v_i$ for every $i \leq m$.
	\end{itemize}
\end{fact}

\begin{proof}
	The second property is an immediate corollary of the first one. Hence, it suffices to show that if $u_1 \cdots u_m$ is a prefix of $v_1 \cdots v_n$, then $m \leq n$ and $u_i = v_i$ for every $i \leq m$. We proceed by induction on $m \in \nat$. If $m= 0$, then $u_1 \cdots u_m = \veps$ and the property is immediate. Otherwise, $m \geq 1$. Clearly, $u_1 \cdots u_{m-1}$ is a prefix of $v_1 \cdots v_n$. Hence, induction yields that $m-1 \leq n$ and $u_i = v_i$ for every $i \leq m -1$. Therefore, since $u_1 \cdots u_m$ is a prefix of $v_1 \cdots v_n$, it follows that $u_m$ is a prefix of $v_m \cdots v_n$. Since $u_m \in K$ which is a prefix code, we have $u_m \neq \veps$ which implies that $v_m \cdots v_n \neq \veps$, i.e. $n \geq m$. Moreover, since $K$ is a prefix code an $u_m,v_m \in K$ we know that $u_m$ is \textbf{not} a strict prefix of $v_m$ and $v_m$ is \textbf{not} a strict prefix of $u_m$. Together with the hypothesis that $u_m$ is a prefix of $v_m \cdots v_n$, this yields $u_m = v_m$, concluding the proof.
\end{proof}

The second assertion in Fact~\ref{fct:sfclos:prefdec} implies that when $K$ is a prefix code, every word $w \in K^*$ admits a \emph{unique} decomposition witnessing this membership. This property is exactly the definition of ``codes'', which are more general than prefix codes.

We have the two following simple facts that we shall use to build new prefix codes with bounded synchronization delay out of already existing ones.

\begin{fact} \label{fct:sfclos:sprefsub}
	Let $d \geq 1$ and $K \subseteq A^+$ a prefix code with synchronization delay $d$. Then, every $H \subseteq K$ is also a prefix code with synchronization delay $d$.
\end{fact}

\begin{proof}
	It is immediate from the definitions that $H$ remains a prefix code. Hence, we need to prove that $H$ has synchronization delay $d$. Consider $u,v,w \in A^*$ such that $uvw \in H^+$ and $v \in H^d$. We show that $uv \in H^+$. Since $H \subseteq K$, we have $uvw \in K^+$ and $v \in K^d$. Hence, since $K$ has synchronization delay $d$, we obtain that $uv \in K^+$. Moreover, since $K$ is a prefix code and $uvw \in K^+$, $uvw$ admits a unique decomposition into factors of $K$ (this is follows from the second property in Fact~\ref{fct:sfclos:prefdec}). Additionally, since $uvw \in H^+$ with $H \subseteq K$, all factors in this unique decomposition belong to $H$. Finally, since $uv\in K^+$, the first property in Fact~\ref{fct:sfclos:prefdec} yields that $uv$ is a concatenation of factors in this unique decomposition. Hence, we have $uv \in H^+$ which concludes the proof.
\end{proof}

Additionally, we have the following more involved construction.

\begin{fact} \label{fct:sfclos:sprefnest}
	Let $d \geq 1$ and $K \subseteq A^+$ a prefix code with synchronization delay $d$. Let $H \subseteq K$. Then, $(K \setminus H)^*H$ is a prefix code with synchronization delay $d+1$. 
\end{fact}

\begin{proof}
	We first verify that $(K \setminus H)^*H$ is a prefix code. Clearly, $(K \setminus H)^*H \subseteq A^+$ since $H \subseteq K \subseteq A^+$. Hence, we have to show that $(K \setminus H)^*HA^+ \cap (K \setminus H)^*H = \emptyset$. If $w$ belongs to this intersection, in particular, we have $w \in K^*$. Therefore, since $K$ is a prefix code, $w$ admits a \emph{unique} decomposition $w = w_1 \cdots w_n$ with $w_1,\dots,w_n \in K$. Since $w \in (K \setminus H)^*H$, the factor $w_n$ is the only one in $H$. However, since $w \in (K \setminus H)^*HA^+$, the first property in Fact~\ref{fct:sfclos:prefdec} implies that one of the factors $w_i$ for $i \leq n-1$ must belong to $H$. This is a contradiction.
	
	It remains to show that $(K \setminus H)^*H$ has synchronization delay $d+1$. We let $u,v,w \in A^*$ such that $uvw \in ((K \setminus H)^*H)^+$ and $v \in ((K \setminus H)^*H)^{d+1}$. We prove that $uv \in ((K \setminus H)^*H)^+$. Clearly $v = xy$ with $x \in ((K \setminus H)^*H)^{d}$ and $y \in (K \setminus H)^*H$. Observe that $x \in K^{n}$ for some $n \geq d$.  Hence, since $uxyw = uvw \in K^+$ and $K$ has synchronization delay $d$, it follows that $ux \in K^+$. Consequently $uv = uxy \in K^+(K \setminus H)^*H$. This implies that $uv \in K^*H$ which is clearly a subset of $((K \setminus H)^*H)^+$. This concludes the proof.
\end{proof}

We may now present the operation $\Cs \mapsto \bsdp{\Cs}$. The definition involves \emph{unambiguous concatenation}. Given $K,L \subseteq A^*$, their concatenation $KL$ is \emph{unambiguous} when every word $w \in KL$ admits a \emph{unique} decomposition $w = uv$ with $u \in K$ and $v \in L$. Given a class \Cs, \bsdp{\Cs} is the least class containing $\emptyset$ and $\{a\}$ for every $a \in A$, and closed under the following properties:
\begin{itemize}
	\item \textbf{Intersection with \Cs:} if $K \in \bsdp{\Cs}$ and $L \in \Cs$, then $K \cap L \in \bsdp{\Cs}$.
	\item \textbf{Disjoint union:} if $K,L \in \bsdp{\Cs}$ are disjoint, then $K \uplus L \in \bsdp{\Cs}$.
	\item \textbf{Unambiguous product:} if $K,L \in \bsdp{\Cs}$ and $KL$ is unambiguous, then $KL \in \bsdp{\Cs}$.
	\item \textbf{Kleene star for prefix codes with bounded synchronization delay:} if $K \in \bsdp{\Cs}$ is a prefix code with bounded synchronization delay, then $K^* \in \bsdp{\Cs}$.
\end{itemize}

\begin{remark}
	Sch\"utzenberger proved in~\cite{schutzbd} that $\bsdp{\{\emptyset,A^*\}} = \sfp{\{\emptyset,A^*\}}$. His definition of \bsdp{\{\emptyset,A^*\}} was slightly less restrictive than ours: it does not require that the unions are disjoint and the concatenations unambiguous. It will be immediate from the correspondence with star-free closure that the two definitions are equivalent. \end{remark}

\begin{remark}
  This closure operation is different from standard ones. Instead of requiring that $\Cs \subseteq \bsdp{\Cs}$, we impose a stronger requirement: intersection with languages in \Cs is allowed. If we only required that $\Cs \subseteq \bsdp{\Cs}$, we would get a weaker operation which does not correspond to star-free closure in general. For example, let $A = \{a,b\}$ and consider the class \md of Example~\ref{ex:allgroups}. Observe $(aa)^* \in \bsdp{\md}$. Indeed, $\{a\} \in \bsdp{\md}$ has bounded synchronization delay, $(AA)^* \in \md$ and $(aa)^* = a^* \cap (AA)^*$. Yet, one may verify that $(aa)^*$ cannot be built from the languages of \md with union, concatenation and Kleene star applied to prefix codes with bounded synchronization delay.
\end{remark}

\subsection{Algebraic characterization of star-free closure}

We now reduce deciding membership for \sfp{\Cs} to computing \Cs-\emph{stutters}. Let us first define this new notion. Let \Cs be a \vari and $\alpha: A^* \to M$ be a morphism. A \emph{\Cs-stutter} for $\alpha$ is an element $s \in M$ such that for every \Cs-cover \Kb of $\alpha\inv(s)$, there exists $K\in\Kb$ satisfying $K\cap KK \neq \emptyset$. When $\alpha$ is understood, we simply speak of \Cs-stutter. Finally, we say that $\alpha$ is \emph{\Cs-aperiodic} when for every \Cs-stutter $s \in M$, we have $s^\omega = s^{\omega+1}$. The reduction is stated in the following theorem.

\adjustc{thm:sfclos:carac}
\begin{theorem}
	Let \Cs be a \vari and consider a regular language $L \subseteq A^*$. The following properties are equivalent:
\begin{enumerate}
		\item $L \in \sfp{\Cs}$.
		\item $L \in \bsdp{\Cs}$.
		\item The syntactic morphism of $L$ is \Cs-aperiodic.
	\end{enumerate}
\end{theorem}
\restorec

Naturally, the characterization need not be effective: this depends on \Cs. Deciding whether a morphism is \Cs-aperiodic boils down to computing \Cs-stutters. Yet, this is possible under the hypotheses of Theorems~\ref{thm:sfclos:main} and~\ref{thm:sfclos:gmain}. First, if \Cs is a finite \vari, deciding whether and element is a \Cs-stutter is simple: there are finitely many \Cs-covers and we may check them all. Moreover, if \Cs is a \vari of group languages, the question boils down to \Cs-separation as stated in the next lemma.

\adjustc{lem:sfclos:grpstut}
\begin{lemma} 
	Let \Cs be a \vari of group languages and $\alpha: A^* \to M$~be a morphism. For all $s \in M$, $s$ is a \Cs-stutter if and only if $\{\veps\}$ is \textbf{not} \Cs-separable from~$\alpha\inv(s)$.
\end{lemma}
\restorec

\begin{proof}
	Assume first that $s$ is a \Cs-stutter. We prove that $\{\veps\}$ is not \Cs-separable from $\alpha\inv(s)$: given $L \in \Cs$ such that $\veps \in L$, we show that $L \cap \alpha\inv(s)\neq\emptyset$. Let $\beta: A^*  \to G$ be the syntactic morphism of $L$. Since \Cs is a \vari of group languages, it is standard and simple to verify that $G$ is a group and that every language recognized by $\beta$ belongs to \Cs (see~\cite{pingoodref} for example). Let \Kb be the set of all languages $\beta\inv(g)$ for $g \in G$ which intersect $\alpha\inv(s)$. By definition, \Kb is a \Cs-cover of $\alpha\inv(s)$. Since $s$ is a \Cs-stutter, this yields $g \in G$ such that $\beta\inv(g)\beta\inv(g) \cap \beta\inv(g) \neq \emptyset$ and $\alpha\inv(s) \cap \beta\inv(g) \neq \emptyset$. Clearly, $\beta\inv(g)\beta\inv(g) \cap \beta\inv(g) \neq \emptyset$ implies that $gg = g$ which means that $g = 1_G$ since $G$ is a group. Therefore, $\beta\inv(g) \subseteq L$ since $L$ is recognized by it syntactic morphism $\beta$ and $\veps \in L$. Together with the hypothesis that $\alpha\inv(s) \cap \beta\inv(g) \neq \emptyset$, this implies $L \cap \alpha\inv(s) \neq \emptyset$ which concludes the proof for the first direction.
	
	Conversely, we assume that $\{\veps\}$ is not \Cs-separable from $\alpha\inv(s)$ and show that $s$ is a \Cs-stutter. Let \Kb be a \Cs-cover of $\alpha\inv(s)$. We need to exhibit $K \in \Kb$ such that $K\cap KK \neq \emptyset$. Let $H$ be the union of all languages in \Kb. Clearly, $H \in \Cs$ and $\alpha\inv(s)\subseteq H$. Therefore, since $\{\veps\}$ is not \Cs-separable from $\alpha\inv(s)$, we have $\veps \in H$. (otherwise $A^* \setminus H \in \Cs$ would be a separator). By definition of $H$, it follows that \Kb contains a language $K$ such that $\veps \in K$. It follows that $\veps \in KK$ and we get $K \cap KK \neq \emptyset$ which concludes the proof.	
\end{proof}

Altogether, we obtain the membership part in Theorems~\ref{thm:sfclos:main} and~\ref{thm:sfclos:gmain}: given a \vari which is either finite or made of group languages and with decidable separation, the \sfp{\Cs}-membership problem is decidable. It remains to prove Theorem~\ref{thm:sfclos:carac}.

\subsection{Proof of Theorem~\ref{thm:sfclos:carac}}

We prove Theorem~\ref{thm:sfclos:carac}. We fix an arbitrary \vari \Cs for the proof and show the implications $2 \Rightarrow 1) \Rightarrow 3) \Rightarrow 2)$. 

\subsubsection*{Direction $2) \Rightarrow 1)$}

We show that $\bsdp{\Cs} \subseteq \sfp{\Cs}$. This amounts to proving that \sfp{\Cs} satisfies all properties involved in the definition of \bsdp{\Cs}. In all cases but one, this is immediate. By definition of \sfp{\Cs}, we have $\{a\} \in \sfp{\Cs}$ for every $a \in A$. Moreover, \sfp{\Cs} is closed under union, intersection and concatenation by definition (this includes intersection with languages of \Cs since $\Cs \subseteq \sfp{\Cs}$). Therefore, we concentrate on proving that \sfp{\Cs} is closed under Kleene star when it is applied to a prefix code with bounded synchronization delay. We fix $K \in \sfp{\Cs}$ which is a prefix code with bounded synchronization delay and prove that $K^* \in \sfp{\Cs}$. We let $d$ be the delay and consider the following language $H$:
\[
H = \left(A^*K^d \cap \left(A^* \setminus \left(A^*K^{d+1} \cup \bigcup_{0 \leq h \leq d} K^h\right)\right) \right) A^*
\]
Since $K$ has synchronization delay $d$, we have the following fact (note that this is the only part of the proof where we use the hypothesis that $K$ has synchronization delay $d$).

\begin{fact} \label{fct:sfclos:delay}
	We have $K^* \subseteq A^* \setminus H$.
\end{fact}

\begin{proof}
	Clearly, we have, $K^* \subseteq A^*K^{d+1} \cup \bigcup_{0 \leq h \leq d} K^h$. Therefore, it follows that,
	\[
	A^* \setminus \left(A^*K^{d+1} \cup \bigcup_{0 \leq h \leq d} K^h\right) \subseteq A^* \setminus K^*
	\]
	By definition of $H$, this yields $H \subseteq \left(A^*K^d \cap \left(A^* \setminus K^*\right)\right) A^*$ and we obtain,
	\[
	A^* \setminus \left(\left(A^*K^d \cap \left(A^* \setminus K^*\right)\right) A^*\right) \subseteq A^* \setminus H
	\]
	Moreover, since $K$ has synchronization delay $d$, one may verify that,
	\[
	K^* \subseteq A^* \setminus \left(\left(A^*K^d \cap \left(A^* \setminus K^*\right)\right) A^*\right)
	\]
	Altogether, this yields as desired that $K^* \subseteq A^* \setminus H$.
\end{proof}

Clearly, we have $H \in \sfp{\Cs}$ by definition of \sfp{\Cs} since $K,A^* \in \sfp{\Cs}$. Finally, we define,
\[
G = \left(\bigcup_{0 \leq h \leq d-1} K^h\right) \cup \left(A^*K^d \cap \left(A^* \setminus H\right)\right)
\]	
Again, it is immediate that $G \in \sfp{\Cs}$. We show that $K^* = G$ which concludes the proof. Let us start with the left to right inclusion. Consider $x \in K^*$. If $x \in K^h$ for $h \leq d-1$, it is immediate that $w \in G$ by definition. Otherwise, we have $x \in A^*K^d$ and since $x \in K^*$, we know that $x \in A^* \setminus H$ by Fact~\ref{fct:sfclos:delay} (note that this fact relies on the hypothesis that $K$ has bounded synchronization delay). This implies that $x \in A^*K^d \cap \left(A^* \setminus H\right)$ which yields $x \in G$, finishing the proof for this inclusion.

\medskip

We finish with the right to left inclusion (which is independent from the hypothesis that $K$ has bounded synchronization delay). Consider $x \in G$, we show that $x \in K^*$. If $x\in \bigcup_{0 \leq h \leq d-1} K^h$, this is immediate. Otherwise, $x\in A^*K^d \cap \left(A^* \setminus H\right)$ and we proceed by induction on the length of $x$. By hypothesis, $x \in A^*K^d$ and $x \not\in H$. By definition of $H$, this implies that,
\[
x \in A^*K^{d+1} \cup \bigcup_{0 \leq h \leq d} K^h
\]
If $x \in \bigcup_{0 \leq h \leq d} K^h$, it is immediate that $x \in K^*$ and we are finished. Otherwise, $x \in A^*K^{d+1}$ which means that $x = x'y$ with $x' \in A^*K^d$ and $y \in K$. Since $\varepsilon \not\in K$ ($K$ is a prefix code), we have $y \neq \veps$ which implies that $|x'| < |x|$. Moreover, since $x \in A^* \setminus H$ and $x'$ is a prefix of $x$, one may verify from the definition of $H$ that $x' \in A^* \setminus H$ as well. Altogether, we have $x' \in A^*K^d \cap \left(A^* \setminus H\right)$ which satisfies $|x'| < |x|$. Therefore, induction yields that $x' \in K^*$. Finally, since $y \in K$, we get $x = x'y \in K^*K \subseteq K^*$ which concludes the proof. 	

\subsubsection*{Direction $1) \Rightarrow 3)$}

Consider a regular language $L \in \sfp{\Cs}$ and let $\alpha: A^* \to M$ be its syntactic morphism. We prove that $\alpha$ is \Cs-aperiodic. First, we show that one may assume without loss of generality that \Cs is \emph{finite}.

\begin{fact} \label{fct:sfclos:profinite}
	There exists a finite \vari \Ds such that $\Ds \subseteq \Cs$ and $L \in \sfp{\Ds}$.
\end{fact}

\begin{proof}
  Since $L \in \sfp{\Cs}$, it is built from finitely many languages in \Cs using the operations available in star-free closure. We let \Ds be the least \vari containing these languages. One may verify that \Ds is finite (this is because a regular language has finitely many quotients by the Myhill-Nerode theorem, see Lemma~17 in~\cite{pzgenconcat} for a proof). Moreover, it is immediate that $\Ds \subseteq \Cs$ and $L \in \sfp{\Ds}$ by definition.
\end{proof}

We fix the finite \vari \Ds described in Fact~\ref{fct:sfclos:profinite}. Moreover, we consider the associated canonical equivalence \caned over $A^*$. We have the following lemma.

\begin{lemma} \label{lem:sfclos:aper}
	Consider $K \in \sfp{\Ds}$. There exists $k \in \nat$ such that for every $\ell \geq k$ and $u,x,y \in A^*$ satisfying $uu \caned u$, we have $xu^\ell y \in K$ if and only if $xu^\ell uy \in K$.
\end{lemma}

\begin{proof}
	By definition, $K$ is built from languages in \Ds  and the singletons $\{a\}$ for $a \in A$ using union, complement and concatenation. We proceed by induction on this construction.
	
	Assume first that $K \in \Ds$. In that case, the proposition holds for $k = 1$. Indeed, if $uu \canec u$, then $xu^{\ell}uy \canec xu^{\ell}y$ for every $\ell \geq 1$ and $x,y \in A^*$ since \canec is a congruence. Therefore, since $K\in \Ds$, it follows that $xu^\ell y \in K$ if and only if $xu^\ell uy \in K$ by definition of \canec. If $K = \{a\}$, one may verify that the property holds for $k = 2$.
	
	We turn to the inductive cases. Assume first that the last operation used to build $K$ is union. We have $K = K_1 \cup K_2$ where $K_1,K_2 \in \sfp{\Ds}$ are simpler languages. Induction yields $k_1,k_2 \in \nat$ such that for $i = 1,2$, if $\ell \geq k_i$ and $u,x,y \in A^*$ satisfy $uu\canec u$, we have $xu^\ell y \in K_i$ if and only if $xu^\ell uy \in K_i$. It is immediate that in this case, the proposition holds for $k$ as the maximum between $k_1$ and $k_2$. We turn to complement. Assume that $K = A^* \setminus H$ where $H \in \sfp{\Ds}$ is a simpler language. Induction yields $h \in \nat$ such that, if $\ell \geq h$ and $u,x,y \in A^*$ satisfy $uu\canec u$, we have $xu^\ell y \in H$ if and only if $xu^\ell u y\in H$. It is immediate that in this case, the proposition holds for $k = h$.
	
	Finally, we assume that the last operation used to construct $K$ is  concatenation. We have $K = K_1K_2$ with $K_1,K_2 \in \sfp{\Ds}$ are simpler languages. Induction yields $k_1,k_2 \in \nat$ such that for $i = 1,2$, if $\ell \geq k_i$ and $u,x,y \in A^*$ satisfy $uu\canec u$, we have $xu^\ell y\in K_i$ if and only if $xu^\ell uy \in K_i$. Let $m$ be the maximum between $k_1$ and $k_2$. We prove that the proposition holds for $k = 2m + 1$. Let $\ell \geq k$ and $u,x,y \in A^*$ such that $uu \canec u$. We need to show that $xu^\ell y \in K$ if and only if $xu^\ell uy\in K$. We concentrate on the right to left direction (the converse one is symmetrical). Assuming that $xu^\ell uy \in K$, we show that $xu^\ell y \in K$. Since $K = K_1K_2$, we get $w_1 \in K_1$ and $w_2 \in K_2$ such that $xu^{\ell} uy = w_1w_2$. Since $\ell \geq 2m + 1$, it follows that either $xu^{m}u$ is a prefix of $w_1$ or $u^{m}uy$ is a suffix of $w_2$. By symmetry we assume that the former property holds: we have $w_1 = xu^{m}uz$ for some $z \in A^*$. Observe that since $xu^{\ell}uy = w_1 w_2$, it follows that $zw_2 =u^{\ell-m}y$. Moreover, we have $m \geq k_1$ by definition of $m$. Therefore, since $xu^{m} uz = w_1 \in K_1$, we know that $xu^{m}z \in K_1$ by definition of $k_1$. Thus, $xu^{m}zw_2 \in K_1K_2 = K$. Since $zw_2 =u^{\ell-m}y$, this yields $xu^{\ell}y \in K$, concluding the proof.	
\end{proof}

We are ready to prove that the syntactic morphism $\alpha: A^* \to M$ of $L$ is \Ds-aperiodic. We fix a \Cs-stutter $s \in M$ for the proof and show that $s^\omega = s^{\omega+1}$. Let \Kb be the set containing all \caned-classes intersecting $\alpha\inv(s)$. Clearly, \Kb is a \Ds-cover of $\alpha\inv(s)$ and therefore a \Cs-cover as well since $\Ds \subseteq \Cs$. Hence, since $s$ is a \Cs-stutter, there exists $K \in \Kb$ such that $K \cap KK \neq \emptyset$.

By definition $K$ is a \caned-class. Moreover, \caned is a congruence for concatenation. Therefore, $K \cap KK \neq \emptyset$ implies that $KK \subseteq K$. Moreover, since $\alpha\inv(s) \cap K \neq \emptyset$ by definition of \Kb, there exists some word $u \in \alpha\inv(s) \cap K$. We have $\alpha(u) = s$. Since $K$ is a \caned-class such that $KK \subseteq K$, we have $uu \caned u$. Therefore, since $K \in \sfp{\Ds}$, Lemma~\ref{lem:sfclos:aper} yields a number $k \in \nat$ such that for every $x,y \in A^*$ and every $\ell \geq k$, we have $xu^\ell y \in L$ if and only if $xu^\ell uy \in L$. In particular, this holds for $\ell = k\omega$ (where $\omega$ is the idempotent power of $M$). Hence, we have $xu^{k\omega} y \in L$ if and only if $xu^{k\omega} uy \in L$. This exactly says that $u^{k\omega}$ and $u^{k\omega} u$ are equivalent for the syntactic congruence of $L$. Consequently, since $\alpha$ is the syntactic morphism of $L$, we have $\alpha(u^{k\omega}) = \alpha(u^{k\omega} u)$. Since $\alpha(u) = s$, this exactly says that $s^\omega = s^{\omega+1}$, concluding the proof.

\subsubsection*{Direction $3) \Rightarrow 2)$}

Consider a \Cs-aperiodic morphism $\alpha: A^* \to M$. We show that every language recognized by $\alpha$ belongs to \bsdp{\Cs}. We start with two definitions. 

Given $K \subseteq A^*$ and $s\in M$, we say that $K$ is \emph{$s$-safe} when $s\alpha(u)=s\alpha(v)$ for every $u,v \in K$. We extend this notion to sets of languages: such a set \Kb is $s$-safe when every $K \in \Kb$ is $s$-safe. Finally, given a language $P \subseteq A^*$, an \bsdp{\Cs}-partition of $P$ is a finite partition of $P$ into languages of \bsdp{\Cs}. The argument is based on the following proposition.

\adjustc{prop:sfclos:fromapertostar}
\begin{proposition}
	Let $P \subseteq A^+$ be a prefix code with bounded synchronization delay. Assume that there exists a $1_M$-safe \bsdp{\Cs}-partition of $P$. Then, for every $s \in M$, there exists an $s$-safe \bsdp{\Cs}-partition of $P^*$.
\end{proposition}
\restorec

We first apply Proposition~\ref{prop:sfclos:fromapertostar} to conclude the main argument. We show that every language recognized by $\alpha$ belongs to \bsdp{\Cs}. By definition, \bsdp{\Cs} is closed under disjoint union. Hence, it suffices to show that $\alpha\inv(t) \in \bsdp{\Cs}$ for every $t \in M$. We fix $t \in M$ for the proof.

Clearly, $A \subseteq A^+$ is a prefix code with bounded synchronization delay and  $\{\{a\} \mid a \in A\}$ is a $1_M$-safe \bsdp{\Cs}-partition of $A$. Hence, Proposition~\ref{prop:sfclos:fromapertostar} (applied in the case $s= 1_M$) yields a $1_M$-safe \bsdp{\Cs}-partition \Kb of $A^*$. One may verify that $\alpha\inv(t)$ is the disjoint union of all $K \in \Kb$ intersecting $\alpha\inv(t)$. Hence, $\alpha\inv(t) \in \bsdp{\Cs}$ which concludes the main argument.

\medskip

It remains to prove Proposition~\ref{prop:sfclos:fromapertostar}. We let $P \subseteq A^*$ be a prefix code with bounded synchronization delay and \Hb be a $1_M$-safe \bsdp{\Cs}-partition of $P$. Moreover, we fix $s \in M$. We need to build an \bsdp{\Cs}-partition \Kb of $P^*$ such that every $K \in \Kb$ is $s$-safe. We proceed by induction on the three following parameters listed by order of importance: $(1)$ the size of $\alpha(P^+) \subseteq M$, $(2)$ the size of $\Hb$ and $(3)$ the size of $s \cdot \alpha(P^*) \subseteq M$.

We distinguish two cases depending on the following property of $s$ and $\Hb$. We say that \emph{$s$ is \Hb-stable} when the following holds:
\begin{equation}\label{eq:sfclos:mostable1}
\text{for every $H \in \Hb$,} \quad  s \cdot \alpha(P^*) = s \cdot \alpha(P^*H).
\end{equation}
The base case happens when $s$ is \Hb-stable. Otherwise, we use induction on our parameters.

\smallskip
\noindent
\textbf{Base case: $s$ is \Hb-stable.} Since $\alpha$ is \Cs-aperiodic, we have the following simple fact.

\adjustc{fct:caracfin}
\begin{fact}
	There is a \emph{finite} \vari $\Ds \subseteq \Cs$ such that $\alpha$ is \Ds-aperiodic.
\end{fact}
\restorec

\begin{proof}
	By definition, for every $s \in M$ which is \textbf{not} a \Cs-stutter, there exists a \Cs-cover $\Kb_s$ of $\alpha\inv(s)$ such that $K\cap KK =\emptyset$ for every $K \in \Kb_s$. Let \Hb be the union of all sets $\Kb_s$. Since \Hb is a finite set of languages in \Cs, one may verify that there exists a finite \vari $\Ds \subseteq \Cs$ containing every $H \in \Hb$ (see Lemma~17 in~\cite{pzgenconcat} for a proof). It is now immediate that $\alpha$ is \Ds-aperiodic.
\end{proof}

Since \Ds is finite, we may consider the associated canonical equivalence \caned over $A^*$. We let $\Kb = \{P^* \cap D \mid D \in \dclac\}$. Clearly, \Kb is a partition of $P^*$. Let us verify that it only contains languages in \bsdp{\Cs}. We have $P \in \bsdp{\Cs}$: it is the disjoint union of all languages in the \bsdp{\Cs}-partition \Hb of $P$. Moreover, $P^* \in \bsdp{\Cs}$ since $P$ is a prefix code with bounded synchronization delay. Hence, $P^* \cap D \in \bsdp{\Cs}$ for every $D \in \dclac$ since $D \in \Ds \subseteq \Cs$. Therefore, it remains to show that every language $K \in \Kb$ is $s$-safe. This is a consequence of the following lemma which is proved using the hypothesis~\eqref{eq:sfclos:mostable1} that $s$ is \Hb-stable.

\adjustc{lem:sfclos:godcase1}
\begin{lemma}
	For every $u,v \in P^*$ such that $u \caned v$, we have $s \alpha(u) = s \alpha(v)$.
\end{lemma}
\restorec

\begin{proof}
	We first use the hypothesis that $s$ is \Hb-stable to prove the following fact.
	
	\begin{fact} \label{fct:sfclos:icarbase}
		Let $q,e \in \alpha(P^*)$ such that $e$ is idempotent. Then, we have $sqe = sq$.
	\end{fact}
	
	\begin{proof}
		The proof is based on the following preliminary result. For every $x,y \in P^*$, we show that there exists $r\in \alpha(P^*)$ such that $sr\alpha(x)=s\alpha(y)$. Since $x \in P^*$, there exists a decomposition $x = x_1 \cdots x_n$ with $x_1,\dots,x_n \in P$. We proceed by induction on the length $n$ of this decomposition. If $n = 0$, then $x = \veps$ and it suffices to choose $r = \alpha(y) \in \alpha(P^*)$. Otherwise, $x = wx'$ with $w \in P$, $x' \in P^*$ admitting a decomposition of length $n-1$. Induction yields $r' \in \alpha(P^*)$ such that $sr'\alpha(x') = s\alpha(y)$. Moreover, since $w \in P$ and \Hb is a partition of $P$, there exists some $H \in \Hb$ such that $w \in H$. Since $r' \in \alpha(P^*)$, we may then use the hypothesis that $s$ is \Hb-stable to obtain $r \in \alpha(P^*)$ and $t \in \alpha(H)$ such that $sr' = srt$. Finally, we know that \Hb is $1_M$-safe by hypothesis. Hence, since $t,\alpha(w) \in \alpha(H)$, we have $t = \alpha(w)$ and $sr'=sr\alpha(w)$. Altogether, this yields $sr \alpha(x) = sr \alpha(w)\alpha(x') = sr'\alpha(x') = s\alpha(y)$ which concludes the proof of our preliminary result.
		
		It remains to prove the fact. Consider $q,e \in \alpha(P^*)$ such that $e$ is idempotent. By definition, we have $x,y \in P^*$ such that $q = \alpha(y)$ and $e = \alpha(x)$. Hence, we have $r \in \alpha(P^*)$ such that $sre = sq$. Finally, since $e$ is idempotent, we obtain that $sqe = sree = sre = sq$ which completes the proof.
	\end{proof}
	
	We may now prove the lemma.	Consider $u,v \in P^*$ such that $u \caned v$. We show that $s \alpha(u) = s \alpha(v)$. We first consider the special case when $u \caned uu$.
	
	Assume that $u \caned uu$. We show that $s \alpha(u)$ and $s \alpha(v)$ are both equal to $s$. Since \caned is a congruence and $u \caned v$, we also have $v \caned vv$. Hence, it suffices to use the hypothesis that $u \caned uu$ to show $s \alpha(u) = s$ (the same result is obtained for $v$ by symmetry). Observe that $\alpha(u)$ must be a \Ds-stutter. Indeed, if \Kb is a \Ds-cover of $\alpha\inv(\alpha(u))$, there exists $K \in \Kb$ such that $u \in K$. Hence, $uu \in K$ since $K\in \Ds$ and $u\caned uu$. Thus, $K \in \Kb$ satisfies $K \cap KK \neq \emptyset$ ($uu$ is in the intersection). By Fact~\ref{fct:caracfin}, $\alpha$ is \Ds-aperiodic. Therefore, $\alpha(u)$ being a \Ds-stutter implies that $(\alpha(u))^\omega = (\alpha(u))^{\omega} \alpha(u)$. We may multiply by $s$ on the left to get $s(\alpha(u))^\omega = s(\alpha(u))^{\omega} \alpha(u)$. Moreover, since $(\alpha(u))^\omega$ is an idempotent of $\alpha(P^*)$, it follows from Fact~\ref{fct:sfclos:icarbase} that $s = s(\alpha(u))^\omega$. Altogether, this yields $s\alpha(u) = s$, concluding this case.
	
	It remains to handle the case when $u$ is not necessarily equivalent to $uu$ for \caned. Since \caned is a congruence, the quotient set \dclac is a finite monoid and it is standard that there exists a number $p \geq 1$ such that $u^p \caned u^pu^p$ (i.e. the \caned-class of $u^p$ is an idempotent) and $\alpha(u^p) \in \alpha(P^*)$ is idempotent. Moreover, since $u \caned v$, we have $u^p \caned vu^{p-1}$. Hence, since $u^p \caned u^pu^p$, it follows from the special case treated above that $s\alpha(u^p) = s\alpha(vu^{p-1})$. Moreover, we may multiply by $\alpha(u)$ on the right side which yields, $s\alpha(u)\alpha(u^p) = s\alpha(v)\alpha(u^{p})$. Finally, since $\alpha(u^p)$ is idempotent, it follows from Fact~\ref{fct:sfclos:icarbase} that $s\alpha(u)\alpha(u^p) = s\alpha(u)$ and $s\alpha(v)\alpha(u^p) = s \alpha(v)$. Altogether, we obtain that $s \alpha(u) = s \alpha(v)$, concluding the proof.
\end{proof}

\smallskip
\noindent
\textbf{Inductive step: $s$ is not \Hb-stable.} By hypothesis, we know that~\eqref{eq:sfclos:mostable1} does not hold. Therefore, we get some $H \in \Hb$ such that the following \textbf{strict} inclusion holds,
\begin{equation}\label{eq:sfclos:godinduca}
s \cdot \alpha(P^*H) \subsetneq s \cdot \alpha(P^*).
\end{equation}
  We fix this language $H \in \Hb$ for the remainder of the proof. The following lemma is proved by induction on our second parameter (the size of $\Hb$).

\adjustc{lem:sfclos:alphind}
\begin{lemma}
	There exists a $1_M$-safe \bsdp{\Cs}-partition \Ub of $(P \setminus H)^*$.
\end{lemma}
\restorec

\begin{proof}
	By Fact~\ref{fct:sfclos:sprefsub}, $P \setminus H$ remains a prefix code with bounded synchronization delay since it is included in $P$. Moreover, it is immediate that $\Gb = \Hb \setminus \{H\}$ is a \bsdp{\Cs}-partition of $P \setminus H$ such that every $G \in \Gb$ is $1_M$-safe. Additionally, it is clear that $\alpha((P \setminus H)^+) \subseteq \alpha(P^+)$ (our first induction parameter has not increased) and $\Gb \subsetneq \Hb$ (our second parameter has decreased). Hence, we may apply induction in Proposition~\ref{prop:sfclos:fromapertostar} for the case when $P,\Hb$ and $s$ have been replaced by $P \setminus H,\Gb$ and $1_M$. This yields a $1_M$-safe \bsdp{\Cs}-partition \Ub of $(P \setminus H)^*$.
\end{proof}

We fix the partition \Ub of $(P \setminus H)^*$ given by Lemma~\ref{lem:sfclos:alphind} and distinguish two independent sub-cases. Since $H \subseteq P$ (as $H$ is an element of the partition \Hb of $P$), we have $\alpha(P^*H) \subseteq \alpha(P^+)$. We use a different argument depending on whether this inclusion is strict or not.

\smallskip
\noindent
\textbf{Sub-case~1: $\alpha(P^*H) = \alpha(P^+)$.} Since $H$ is $1_M$-safe by hypothesis, there exists $t \in M$~such that $\alpha(H) = \{t\}$. Similarly, since every $U \in \Ub$ is $1_M$-safe, there exists $r_U \in M$ such that $\alpha(U) = \{r_U\}$. The construction of \Kb is based on the next lemma which is proved using~\eqref{eq:sfclos:godinduc}, the hypothesis of Sub-case~1 and induction on our third parameter (the size of $s \cdot \alpha(P^*) \subseteq M$).

\adjustc{lem:sfclos:sc1carac}
\begin{lemma}
	 For every $U \in \Ub$, there exists an $sr_Ut$-safe \bsdp{\Cs}-partition $\Wb_{U}$ of $P^*$.
\end{lemma}
\restorec

\begin{proof}
	We fix $U \in \Ub$ for the proof. Since \Ub is a partition of $(P \setminus H)^*$, we have $\alpha(U) \subseteq \alpha(P^*)$ which means that $r_U \in \alpha(P^*)$. Thus, we have $sr_Ut \in s\alpha(P^*H)$. Therefore, $sr_Ut\alpha(P^*) \subseteq s\alpha(P^*HP^*)$ and since $H \subseteq P$, we get $sr_Ut\alpha(P^*) \subseteq s\alpha(P^+)$. Combined with our hypothesis in Sub-case~1 (i.e. $\alpha(P^*H) = \alpha(P^+)$), this yields $sr_Ut\alpha(P^*) \subseteq s\alpha(P^*H)$. Finally, we obtain from~\eqref{eq:sfclos:godinduca} (i.e. $s\alpha(P^*H) \subsetneq s\alpha(P^*)$) that the \textbf{strict} inclusion $sr_Ut\alpha(P^*) \subsetneq s\alpha(P^*)$ holds. Consequently, by induction on our third parameter (i.e. the size of $s\alpha(P^*)$) we may apply Proposition~\ref{prop:sfclos:fromapertostar} in the case when $s \in M$ has been replaced by $sr_Ut \in M$. Note that here, our first two parameters have not increased (they only depend on $P$ and \Hb which remain unchanged). This yields the desired \sfp{\Cs}-partition $\Wb_{U}$ of $P^*$.
\end{proof}

  We are ready to define the partition \Kb of $P^*$. Using Lemma~\ref{lem:sfclos:sc1carac}, we define,
\[
\Kb = \Ub \cup \bigcup_{U \in \Ub} \{UHW \mid W \in \Wb_U\}
\]
It remains to show that \Kb is an $s$-safe \bsdp{\Cs}-partition of~$P^*$. First, \Kb is a partition of~$P^*$ since $P$ is a prefix code and $H \subseteq P$. Indeed, every word $w \in P^*$ admits a \emph{unique} decomposition $w = w_1 \cdots w_n$ with $w_1,\dots,w_n \in P$. If no factor $w_i$ belongs to $H$, then $w \in (P \setminus H)^*$ and $w$ belongs to some unique $U \in \Ub$. Otherwise, let $w_i$ be the leftmost factor such that $w_i \in H$. Thus, $w_1 \cdots w_{i-1} \in (P \setminus H)^*$, which also yields a unique $U \in \Ub$ such that $w_1 \cdots w_{i-1} \in U$ and $w_{i+1} \cdots w_n \in P^*$ which yields a unique $W \in \Wb_U$ such that $w_{i+1} \cdots w_n \in W$. It follows that $w \in UHW$ which is an element of \Kb (and the only one containing $w$).

Moreover, every $K \in \Kb$ belongs to \bsdp{\Cs}. If $K \in \Ub$, this is immediate by definition of \Ub in Lemma~\ref{lem:sfclos:alphind}. Otherwise, $K = UHW$ with $U \in \Ub$ and $W \in \Wb_U$. We know that $U,H,W \in \bsdp{\Cs}$ by definition. Moreover, one may verify that the concatenation $UHW$ is \emph{unambiguous} since $P$ is a prefix code, $U \subseteq (P \setminus H)^*$ and $W \subseteq H^*$. Hence, $K \in \bsdp{\Cs}$.

Finally, we verify that \Kb is $s$-safe. Consider $K \in \Kb$ and $w,w' \in K$, we show that $s\alpha(w) = s\alpha(w')$. If $K \in \Ub$, this is immediate: \Ub is $1_M$-safe by definition. Otherwise, $K=UHW$ with $U \in \Ub$ and $W \in \Wb_U$. By definition, $\alpha(H) = \{t\}$ and $\alpha(U) = \{r_U\}$ which implies that $s\alpha(w) = str_U \alpha(x)$ and $s\alpha(w') = str_U \alpha(x')$ for $x,x' \in W$. Moreover, $W \in \Wb_U$ is $sr_Ut$-safe by definition. Hence, $s\alpha(w) = s \alpha(w')$, which concludes the proof of this sub-case.

\smallskip
\noindent
\textbf{Sub-case~2: we have the strict inclusion $\alpha(P^*H) \subsetneq \alpha(P^+)$.} Consider $w \in P^*$. Since $P$ is a prefix code, $w$ admits a unique decomposition $w = w_1 \cdots w_n$ with $w_1,\dots,w_n \in P$. We may look at the rightmost factor $w_i \in H \subseteq P$ to uniquely decompose $w$ in two parts (each of them possibly empty): the prefix $w_1 \cdots w_i \in ((P \setminus H)^*H)^*$ and the suffix in $w_{i+1} \cdots w_n \in (P \setminus H)^*$. Using induction,~we construct \bsdp{\Cs}-partitions of the possible languages of prefixes and suffixes. Then, we combine them to construct a partition of the whole set $P^*$. We already handled the suffixes: \Hb is an \bsdp{\Cs}-partition of $(P \setminus H)^*$. The prefixes are handled with the next lemma, whose proof uses the hypothesis of Sub-case~2 and induction on our first parameter (the size of~$\alpha(P^+)$).

\adjustc{lem:sfclos:sc2carac}
\begin{lemma}
There exists a $1_M$-safe \bsdp{\Cs}-partition \Vb of $((P \setminus H)^*H)^*$.
\end{lemma}
\restorec

\begin{proof}
	Let $Q=(P \setminus H)^*H$. By Fact~\ref{fct:sfclos:sprefnest}, $Q$ remains a prefix code with bounded synchronization delay. We apply induction in Proposition~\ref{prop:sfclos:fromapertostar} for the case when $P$ has been replaced by $Q$. Doing so requires building an appropriate \sfp{\Cs}-partition of $Q$ and proving that one of our induction parameters has decreased.
	
	Let $\Fb = \{UH \mid U \in \Ub\}$. Since \Ub is a partition of $(P\setminus H)^*$ and $P$ is a prefix code, one may verify that \Fb is a partition of $Q = (P\setminus H)^*H$. Moreover, it only contains languages in \bsdp{\Cs}. Indeed, if $U \in \Ub$, then the concatenation $UH$ is \emph{unambiguous} since $U\subseteq (P\setminus H)^*$ and $P$ is a prefix code. Moreover, $U,H \in \bsdp{\Cs}$ by hypothesis. Finally, $UH$ is $1_M$-safe since this is the case for both $U$ and $H$ by definition. It remains to show that our induction parameters have decreased. Since $Q = (P \setminus H)^*H$, it is clear that $Q^+ \subseteq P^*H$. Hence, $\alpha(P^*H) \subsetneq \alpha(P^+)$ by hypothesis in Sub-case~2, we have $\alpha(Q^+) \subsetneq \alpha(P^+)$: our first induction parameter has decreased. Thus, we may apply Proposition~\ref{prop:sfclos:fromapertostar} in the case when $P,\Hb$ and $s$ have been replaced by $Q,\Fb$ and $1_M$. This yields the desired \bsdp{\Cs}-partition \Vb of $((P \setminus H)^*H)^*$.
\end{proof}

 Using Lemma~\ref{lem:sfclos:sc2carac}, we define $\Kb = \{VU \mid V \in \Vb \text{ and } U \in \Ub\}$. It follows from the above discussion that \Kb is a partition of $P^*$ since  \Vb and \Ub are partitions of $((P \setminus H)^*H)^*$ and $(P \setminus H)^*$, respectively. Moreover, every $K \in \Kb$ belongs to \bsdp{\Cs}: $K = VU$ with $V \in \Vb$ and $U\in \Ub$, and one may verify that this is a \emph{unambiguous} concatenation. It remains to show that \Kb is $s$-safe. Let $K \in \Kb$ and $w,w' \in K$. We show that $s\alpha(w) = s\alpha(w')$. By definition, we have $K = VU$ with $V \in \Vb$ and $U \in \Ub$. Therefore, $w =vu$ and $w' = v'u'$ with $u,u' \in U$ and $v,v' \in V$. Since $U$ and $V$ are both $1_M$-safe by definition, we have $\alpha(u) = \alpha(u')$ and $\alpha(v) = \alpha(v')$. It follows that $s\alpha(w) = s\alpha(w')$, which concludes the proof of Proposition~\ref{prop:sfclos:fromapertostar}.

\section{Appendix to Section~\ref{sec:finite}}
\label{app:appcovf}
This is devoted to the proof of Theorem~\ref{thm:sfclos:carfinite}: the characterization of \sfp{\Cs}-optimal \imprints which is generic to all finite \varis \Cs. We fix \Cs for the presentation. First, we recall the theorem and then concentrate on its proof.

\subsection{Characterization}

In view of Proposition~\ref{prop:breduc}, deciding \sfp{\Cs}-covering amounts to computing \sfcopti from an input \nice \mratm $\rho$. Our algorithm actually computes slightly more information.

Since \Cs is a finite \vari, we may consider the equivalence \canec over $A^*$. In particular, the set \sclac of \canec-classes if a finite monoid (we write ``\cmult'' for its multiplication) and the map $w \mapsto \ctype{w}$ is a morphism. Given a \ratm $\rho: 2^{A^*} \to R$ we define:
\[
\csfcopti = \{(C,r) \in (\sclac) \times R \mid r \in \opti{\sfp{\Cs}}{C,\rho}\}
\]
Observe that \csfcopti captures more information than \sfcopti. Indeed, it encodes all sets \opti{\sfp{\Cs}}{C,\rho} for $C \in \sclac$ and by Fact~\ref{fct:lunion}, \sfcopti is the union of all these sets.

Our main result is a least fixpoint procedure for computing \csfcopti from a \nice \mratm $\rho$. It is based on a generic characterization theorem which we first present. Given an arbitrary \nice \mratm $\rho: 2^{A^*} \to R$ and a set $S \subseteq (\sclac) \times R$, we say that $S$ is \emph{\sfp{\Cs}-saturated for $\rho$} when the following properties are satisfied:
\begin{enumerate}
	\item {\bf Trivial elements.} For every $w \in A^*$, we have $(\ctype{w},\rho(w)) \in S$.
	\item {\bf Downset.} For every $(C,r) \in S$ and $q \in R$, if $q \leq r$, then $(C,q) \in S$.
	\item {\bf Multiplication.} For every $(C,q),(D,r) \in S$, we have $(C \cmult D,qr) \in S$.
	\item {\bf \sfp{\Cs}-closure.} For all $(E,r) \in S$, if $E \in \sclac$ is idempotent, then $(E,r^\omega + r^{\omega+1}) \in S$.
\end{enumerate}

\adjustc{thm:sfclos:carfinite}
\begin{theorem}[\sfp{\Cs}-optimal \imprints (\Cs finite)]
	Let $\rho: 2^{A^*} \to R$ be a \emph{\nice} \mratm. Then, \csfcopti is the least \sfp{\Cs}-saturated subset of $(\sclac) \times R$ for $\rho$.
\end{theorem}
\restorec

Given a \nice \mratm $\rho: 2^{A^*} \to R$ as input, it is clear the one may compute the least \sfp{\Cs}-saturated subset of $(\sclac) \times R$ with a least fixpoint procedure. Hence, Theorem~\ref{thm:sfclos:carfinite} provides an algorithm for computing \csfcopti. As we explained above, we may then compute \sfcopti from this set. Together with Proposition~\ref{prop:breduc}, this yields Theorem~\ref{thm:sfclos:main} as a corollary: \sfp{\Cs}-covering is decidable when \Cs is a finite \vari.

\subsection{Proof}

We turn to the proof of Theorem~\ref{thm:sfclos:carfinite}. Let us fix a \nice \mratm $\rho: 2^{A^*} \to R$ for the argument. We prove that \csfcopti is the least \sfp{\Cs}-saturated subset of $(\sclac) \times R$ (for $\rho$). The argument involves two directions which correspond respectively to soundness and completeness of the least fixpoint procedure which computes \csfcopti from $\rho$.
\begin{itemize}
	\item \textbf{Soundness:} We prove that \csfcopti is \sfp{\Cs}-saturated.
	\item \textbf{Completeness:} We prove that every \sfp{\Cs}-saturated set is included in \csfcopti.
\end{itemize}

\subsubsection*{Soundness}

First, we prove that the set $\csfcopti \subseteq (\sclac) \times R$ itself is \sfp{\Cs}-saturated. The argument is based on Lemma~\ref{lem:sfclos:aper}.

\begin{remark}
	We do not need the hypothesis that $\rho$ is \nice for this direction of the proof.
\end{remark}

That \csfcopti contains the trivial elements and is closed under downset and multiplication is actually a generic property of optimal \imprints: this hold as soon as the investigated class is a \vari (see Lemma~9.5 in~\cite{pzcovering2}). This is the case for \sfp{\Cs} since~\Cs is a \vari by hypothesis. Hence, we concentrate on proving that \csfcopti satisfies \sfp{\Cs}-closure. Consider $(E,r) \in \csfcopti$ such that $E \in \sclac$ is idempotent. We show that $(E,r^\omega + r^{\omega+1}) \in \csfcopti$. By definition, this corresponds to the following property:
\[
r^\omega + r^{\omega+1} \in \opti{\sfp{\Cs}}{E,\rho}.
\]
This amounts to proving that for every \sfp{\Cs}-cover \Kb of $E$, we have $r^\omega + r^{\omega+1} \in \prin{\rho}{\Kb}$. We fix \Kb for the proof: we have to exhibit $K \in \Kb$ such that $r^\omega + r^{\omega+1} \leq \rho(K)$. We start with a few definitions that we require to describe $K$. Since \Kb is finite and \sfp{\Cs} is \vari, we have the following fact (see Lemma~17 in~\cite{pzgenconcat} for a proof).

\begin{fact} \label{fct:sfclos:strat}
	There exists a finite \vari \Ds such that $\Ds \subseteq \sfp{\Cs}$ and every $K \in \Kb$ belongs to \Ds.
\end{fact}

We fix \Ds as the finite \vari given by Fact~\ref{fct:sfclos:strat}. Recall that \caned denotes the associated canonical equivalence defined on $A^*$. Since \Ds is closed under quotients we know that \caned is a congruence for word concatenation. Additionally, we have the following more involved property which is a corollary of Lemma~\ref{lem:sfclos:aper}.

\begin{lemma} \label{lem:sfclos:choiceofk}
	There exists a natural number $k \in \nat$ such that for every $\ell \geq k$ and $u \in E$, we have $u^{\ell} \caned u^{\ell} u$.
\end{lemma}

\begin{proof}
	Lemma~\ref{lem:sfclos:aper} yields that for every language $L \in \Ds \subseteq \sfp{\Cs}$, there exists $k_L, \in \nat$ such that for every $\ell \geq k_L$ and every $u \in A^*$ satisfying $uu \canec u$, we have $u^\ell \in L$ if and only if $u^\ell u \in L$. We choose $k$ as the maximum of all numbers $k_L$ for $L \in \Ds$ (recall that \Ds is finite). It remains to show that the lemma holds for this $k$. Consider $\ell \geq k$ and $u \in E$. Since $E$ is an idempotent \canec-class, we have $uu \canec u$. Hence, by choice of $k$, it is immediate that $u^{\ell} \in L$ if and only if $u^{\ell} u \in L$ for every $L \in \Ds$. By definition, this exactly says that $u^{\ell} \caned u^{\ell} u$.
\end{proof}

We may now come back to the main argument. We write \Hb for the set of all $\caned$-classes which intersect $E$. Clearly, \Hb is a \Ds-cover of $E$ and therefore an $\sfp{\Cs}$-cover of $E$ since $\Ds \subseteq \sfp{\Cs}$ by definition in Fact~\ref{fct:sfclos:strat}. Hence, since $(E,r) \in \csfcopti$ by hypothesis (which means that $r \in \opti{\sfp{\Cs}}{E,\rho}$), we obtain that $r \in \prin{\rho}{\Hb}$. This yields $H \in \Hb$ such that $r \leq \rho(H)$. Consider the natural number $k \in \nat$ given by Lemma~\ref{lem:sfclos:choiceofk} and the idempotent power $\omega$ of $R$ for multiplication. We define,
\[
G = H^{k\omega} \cup H^{k\omega+1}
\]
The argument is now based on the following lemma which exhibits the desired language $K \in \Kb$ such that $r^\omega + r^{\omega+1} \leq \rho(K)$.

\begin{lemma} \label{lem:sfclos:findthek}
	The exists $K \in \Kb$ such that $G \subseteq K$.
\end{lemma}

Before we prove Lemma~\ref{lem:sfclos:findthek}, let us use it to conclude the argument. By definition of $G$ and since $r \leq \rho(H)$, we get,
\[
r^\omega + r^{\omega+1} \leq \rho(H^{k\omega} \cup H^{k\omega+1}) = \rho(G)
\]
Furthermore, since the language $K \in \Kb$ given by Lemma~\ref{lem:sfclos:findthek} satisfies $G \subseteq K$, we obtain that $r^\omega + r^{\omega+1} \leq \rho(K)$ as desired. This yields $r^\omega + r^{\omega+1}\in \prin{\rho}{\Kb}$ finishing the proof. We now prove Lemma~\ref{lem:sfclos:findthek}.

\begin{proof}[Proof of Lemma~\ref{lem:sfclos:findthek}]
	Since $H \in \Hb$, we know that $H$ is a $\caned$-class intersecting $E$. Let $u \in H \cap E$ and $v = u^{k\omega}$. Clearly, $v \in E^{k \omega}$ and since $E$ is an idempotent of \sclac, it is immediate that $E^{k \omega} \subseteq E$ which yields $v \in E$. By hypothesis, \Kb is a cover of $E$. Thus, we have $K \in \Kb$ such that $v \in E$. We show that $G \subseteq K$ which concludes the proof.
	
	Consider $w \in G$. We have to prove that $w \in K$. We show that $w \caned v$. Since  $v \in K$ by definition of $K$ and $K \in \Ds$ (see Fact~\ref{fct:sfclos:strat}), this implies that $w \in K$. By definition of $G$, either $w \in H^{k\omega}$ or $w \in H^{k\omega+1}$. We treat the two cases separately. Assume first that $w \in H^{k\omega}$. Recall that $u \in H$ and $H$ is a $\caned$-class. Thus, since $\caned$ is a congruence, it follows that $w \caned u^{k\omega} = v$ which concludes the proof. Assume now that $w \in H^{k\omega+1}$ Again, since $\caned$ is a congruence, we get that $w \caned u^{k\omega+1}$.  Moreover, we have $u \in E$. Therefore, Lemma~\ref{lem:sfclos:choiceofk} yields that $u^{k\omega+1} \caned u^{k\omega}=v$. Transitivity then yields $w \caned v$ which concludes the proof.
\end{proof}

\subsubsection*{Completeness}

We turn to the most interesting direction in Theorem~\ref{thm:sfclos:carfinite}. We show that \csfcopti is included in every \sfp{\Cs}-saturated subset of $(\sclac) \times R$ for $\rho$. Consequently, we fix $S \subseteq (\sclac) \times R$ which is \sfp{\Cs}-saturated for the proof. As expected, the argument involves building a particular \sfp{\Cs}-cover.

We start with some terminology that we need to present the argument. It simplifies the presentation to use ``$(\sclac) \times R$'' as the rating set of a new \ratm $\gamma$ that we build from $\rho$. However, since \sclac is not a semiring, we need to slightly modify this set. Since \sclac is a finite monoid, $2^{\sclac}$ is clearly an idempotent semiring: the addition is union and the multiplication is lifted from the one of \sclac. Hence, $Q = 2^{\sclac} \times R$ is an idempotent semiring for the componentwise addition and multiplication. Consider the map $\gamma: 2^{A^*} \to Q$ defined by $\gamma(K) = (\{\ctype{w} \mid w \in K\},\rho(K))$ for every $K\subseteq A^*$. It is straightforward to verify that $\gamma$ is a \nice \mratm since $\rho$ is one. Moreover, we reformulate our \sfp{\Cs}-saturated set $S \subseteq (\sclac) \times R$ as the following subset $T$ of $Q$:
\[
T = \{(\{C\},r) \mid (C,r) \in S\} \subseteq Q = 2^{\sclac} \times R
\]
Since $S$ is \sfp{\Cs}-saturated for $\rho$, we know that $(\ctype{w},\rho(w)) \in S$ for every $w \in A^*$ (this is a trivial element) and $S$ is closed under multiplication. By definition of $T$, this implies that $\tau(w) \in T$ for every $w \in A^*$ (in particular $1_Q = \tau(\veps) \in T$) and $T$ is closed under multiplication. From now on, we assume that these two properties of $T$ are understood.

\newcommand{\csatp}[1]{\ensuremath{Q^+_{#1}}\xspace}
\newcommand{\csats}[1]{\ensuremath{Q^*_{#1}}\xspace}

Finally, for every finite set of languages \Hb, we associate two subsets of $Q$. The definitions are as follows:
\begin{itemize}
	\item $\csatp{\Hb} \subseteq Q$ is the least subset of $Q$ closed under addition and multiplication such that $\rho(H) \in \csatp{\Hb}$ for every $H \in \Hb$.
	\item $\csats{\Hb} \subseteq Q$ is the least subset of $Q$ closed under addition and multiplication such that $1_Q \in \csats{\Hb}$ and $\rho(H) \in \csats{\Hb}$ for every $H \in \Hb$.
\end{itemize}

We are ready to present the completeness argument. It is based on the following statement which generalizes Proposition~\ref{prop:sfclos:fromapertostar}. We prove it by induction. Recall that given a language $P \subseteq A^*$, a \bsdp{\Cs}-partition of $P$ is a finite partition of $P$ into languages of $\bsdp{\Cs} =\sfp{\Cs}$.

\begin{proposition} \label{prop:sfclos:pumping}
	Let $P \subseteq A^+$ be a prefix code with bounded synchronization delay and \Hb a \bsdp{\Cs}-partition of $P$ such that $\gamma(H) \in T$ for every $H \in \Hb$. Then, for every $t \in T$, there exists a \bsdp{\Cs}-partition \Kb of $P^*$ satisfying the following property,
	\begin{equation} \label{eq:sfclos:covergoal}
	\text{for every $K \in \Kb$, $\gamma(K) \in \csats{\Hb}$ and $t \cdot \gamma(K) \in T$}
	\end{equation}   	
\end{proposition}

Before we prove Proposition~\ref{prop:sfclos:pumping}, let us first apply it to prove the completeness direction in Theorem~\ref{thm:sfclos:carfinite}. We start by constructing a \bsdp{\Cs}-partition \Kb of $A^*$ with the proposition.

Observe that $A \subseteq A^+$ is a prefix code with bounded synchronization delay. Moreover, $\Hb = \{\{a\} \mid a \in A\}$ is a \bsdp{\Cs}-partition of $A$ and we have $\gamma(a) \in T$ for every $a \in A$. Therefore, we may apply Proposition~\ref{prop:sfclos:pumping} in the case when $P = A$ and $t = 1_Q \in T$. This yields a \bsdp{\Cs}-partition \Kb of $A^*$ satisfying~\eqref{eq:sfclos:covergoal}. We use it to prove that $\csfcopti \subseteq S$.

Let $(C,r) \in \csfcopti$, i.e. $r \in \opti{\sfp{\Cs}}{C,\rho}$. We define $\Kb_C \subseteq \Kb$ as the set containing every language $K \in \Kb$ such that $K \cap C \neq \emptyset$. Clearly, $\Kb_C$ is a \bsdp{\Cs}-cover of $C$ since \Kb is a \bsdp{\Cs}-cover of $A^*$. Moreover, it is a \sfp{\Cs}-cover of $C$ since $\sfp{\Cs} = \bsdp{\Cs}$ by Theorem~\ref{thm:sfclos:carac}. It follows that $\opti{\sfp{\Cs}}{C,\rho} \subseteq \prin{\rho}{\Kb_C}$ and we obtain $r \in \prin{\rho}{\Kb_C}$. This yields a language $K \in \Kb_C$ such that $r \leq \rho(K)$. Moreover, $K \in \Kb$ and~\eqref{eq:sfclos:covergoal} yields that $\gamma(K) \in T$ (recall that we chose $t = 1_Q$). By definition of $T$, it follows that $\gamma(K) = (\{D\},p)$ for some $(D,p) \in S$. By definition of $\gamma$, $p = \rho(K)$ and since $K \cap C \neq \emptyset$, we have $D = C$. Hence, $(C,\rho(K)) \in S$ and since $r \leq \rho(K)$, closure under downset for $S$ (recall that $S$ is \sfp{\Cs}-saturated) yields $(C,r) \in S$, concluding the proof.

\medskip

It remains to prove Proposition~\ref{prop:sfclos:pumping}. We let $P \subseteq A^+$ be a prefix code with bounded synchronization delay and \Hb a \bsdp{\Cs}-partition of $P$ such that $\gamma(H) \in T$ for every $H \in \Hb$. Finally, we fix $t \in T$. We need to build a \bsdp{\Cs}-partition \Kb of $P^*$ satisfying~\eqref{eq:sfclos:covergoal}. We proceed by induction on the three following parameters, listed by order of importance:
\begin{enumerate}
	\item The size of the set $\csatp{\Hb} \subseteq Q$,
	\item The size of \Hb,
	\item The size of the set $t \cdot \csats{\Hb} \subseteq Q$.
\end{enumerate}

We distinguish two main cases depending on the following property. We say that \emph{$t$ is $\Hb$-stable} when the following holds, 
\begin{equation} \label{eq:sfclos:ratstable}
\text{for every $H \in \Hb$,} \quad t \cdot \csats{\Hb} = t \cdot \csats{\Hb} \cdot \gamma(H)
\end{equation}
We first consider the case when $t$ is $\Hb$-stable. This is the base case: we construct \Kb directly. Then, we handle the converse case using induction on our three parameters.

\subsubsection*{Base case: $t$ is $\Hb$-stable}

In this case, we define \Kb directly: \Kb contains all languages $P^* \cap C$ for $C \in \sclac$ which are \textbf{nonempty}. Clearly, this is a partition of $P^*$. Moreover, it only contains languages in \bsdp{\Cs}. Indeed, we have $P\in\bsdp{\Cs}$: it is the disjoint union of all languages in the \bsdp{\Cs}-partition \Hb of $P$. Hence, $P^*\in \bsdp{\Cs}$ since $P$ is a prefix code with bounded synchronization delay. Finally, it follows that $P^* \cap C \in \bsdp{\Cs}$ for every $C \in \sclac$ since $C \in \Cs$. We have to prove \Kb satisfies~\eqref{eq:sfclos:covergoal}: given $K \in \Kb$, we show that $\gamma(K) \in \csats{\Hb}$ and $t \cdot \gamma(K) \in T$.

We begin with an important definition. A \emph{\Hb-product} is a language of the form $H_1 \cdots H_n$ with $n \in \nat$ and $H_1,\dots,H_n \in \Hb$ (this includes the case $n=0$: $\{\veps\}$ is a $\Hb$-product). Clearly, $\gamma(V) \in \csats{\Hb}$, for every \Hb-product $V$. Moreover, since $\gamma(H) \in T$ for every $H \in \Hb$ and $T$ is closed under multiplication and contains $1_Q$, we also have $\gamma(V) \in T$ for every \Hb-product $V$. Using this observation, we prove two simple properties of \Hb-products.

\begin{fact} \label{fct:sfclos:mim}
	Every \Hb-product $V$ is nonempty and there exists a unique $C \in \sclac$, denoted by \ctype{V} such that $V \subseteq C$. \end{fact}

\begin{proof}
	Since $\gamma(V) \in T$, we get by definition of $T$ that $\gamma(V) = (\{C\},r)$ for some $C \in \sclac$ and $r \in R$. By definition of $\gamma$, we have $\{C\} = \{\ctype{w} \mid w \in V\}$. Thus, $V$ is nonempty and $V \subseteq C$. 
\end{proof}

We turn to an important property of \Hb-products: they are the building blocks of the languages within our \bsdp{\Cs}-partition \Kb of $P^*$.

\begin{fact} \label{fct:sfclos:images}
	For every $C \in \sclac$, $P^* \cap C$ it the (possibly infinite or empty) union of all \Hb-products $V$ such that $\ctype{V} = C$.
\end{fact}

\begin{proof}
	By definition, \Hb is a partition of $P$. Hence, $P^* \cap C$ is the (infinite) union of all languages $H_1 \cdots H_n \cap C$ for $n \in \nat$ and $H_1,\dots,H_n \in \Hb$. Consider such a language. Clearly, $H_1 \cdots H_n$ is a \Hb-product. Therefore, either $\ctype{H_1 \cdots H_n} = C$ and $H_1 \cdots H_n \cap C = H_1 \cdots H_n$ or $\ctype{H_1 \cdots H_n}\neq C$ and $H_1 \cdots H_n \cap C = \emptyset$. The fact follows.
\end{proof}

We may start the proof that \Kb satisfies~\eqref{eq:sfclos:covergoal}. First, we show that $\gamma(K) \in \csats{\Hb}$ for every $K \in \Kb$. By definition, $K$ is a nonempty language of the form $P^* \cap C$ with $C \in \sclac$. Hence, Fact~\ref{fct:sfclos:images} yields that $K$ is a (possibly infinite) union of \Hb-products. Since $\gamma$ is \nice by definition, it follows that there exists finitely many \Hb-products $V_1,\dots,V_\ell$ such that $\gamma(K) = \gamma(V_1) + \cdots + \gamma(V_\ell)$. We know that $\gamma(V_i) \in \csats{\Hb}$ for every $i \leq \ell$. Therefore, $\gamma(K) \in \csats{\Hb}$ since \csats{\Hb} is closed under addition by definition.

\medskip

It remains to show that $t\gamma(K) \in T$ for every $K \in \Kb$. This is where we need the hypothesis that $t$ is \Hb-stable. We use it to prove the following preliminary lemma.

\begin{lemma} \label{lem:sfclos:basecase}
	Let $V$ be an \Hb-product such that $\gamma(V)$ is a multiplicative idempotent of $Q$. Then, for every $q \in \csats{\Hb}$, we have $tq\gamma(V) = tq$.
\end{lemma}

\begin{proof}	
	We first use the hypothesis that $t$ is $\Hb$-stable to prove the following preliminary result which holds regardless of whether $\gamma(V)$ is idempotent:
	\begin{equation} \label{eq:sfclos:auxproof}
	\text{there exists $x \in \csats{\Hb}$ such that $tx\gamma(V) = tq$}
	\end{equation}
	Since $V$ is a \Hb-product, there exists $n \in \nat$ and $H_1,\dots,H_n \in \Hb$ such that $V = H_1 \cdots H_n$. We proceed by induction on $n$. If $n = 0$, then $V = \{\veps\}$ and $\gamma(V) = 1_Q$, it suffices to choose $x = q \in \csats{\Hb}$. Otherwise $V = HV'$ for $H \in \Hb$ and $V'$ a \Hb-product of smaller length. Using induction, we get $y \in \csats{\Hb}$ such that $ty\gamma(V') = tq$. Moreover, since $t$ is $\Hb$-stable, we get from~\eqref{eq:sfclos:ratstable} that there exists $x \in \csats{\Hb}$ such that $ty = tx\gamma(H)$. It follows that $tx\gamma(V) = tx\gamma(H)\gamma(V') = ty\gamma(V') = tq$. This concludes the proof of~\eqref{eq:sfclos:auxproof}. It remains to prove the lemma.
	
	We have $tx\gamma(V) = tq$. Hence, if $\gamma(V)$ is a multiplicative idempotent of $Q$, we get $tq\gamma(V) = tx\gamma(V)\gamma(V) = tx\gamma(V) = tq$. This completes the proof.
\end{proof}

We are ready to prove that $t\gamma(K) \in T$ for every $K \in \Kb$. We fix $K \in \Kb$ for the proof. By definition, $K$ is non empty and $K = P^* \cap C$ for $C \in \sclac$. We first treat the case when $C$ is an idempotent $E$ of \sclac. Then we use it to treat the general case.

\medskip
\noindent
{\bf Idempotent case.} By Fact~\ref{fct:sfclos:images}, $K =  P^* \cap E$ is the union of all $\Hb$-products $V$ such $\ctype{V}=E$. This union is non-empty since $K$ is nonempty. Since $\gamma$ is \nice, this yields finitely many $\Hb$-products $V_1,\dots,V_\ell$ such that $\gamma(K) = \gamma(V_1) + \cdots + \gamma(V_\ell)$ and $\ctype{V_i} = E$ for every $i\leq \ell$.

Consider $i \leq \ell$. Moreover, $\ctype{V_i} = E$ means that $V_i \neq \emptyset$ and $V_i \subseteq E$ by Fact~\ref{fct:sfclos:mim}. Therefore, $\gamma(V_i)=(\{E\},\rho(V_i))$ by definition of $\gamma$. Moreover, we know that $\gamma(V_i) \in T$ which yields $(E,\rho(V_i)) \in S$ by definition of $T$. Since $E$ is idempotent and $S$ is \sfp{\Cs}-saturated, we obtain from \bsdp{\Cs}-closure that $(E,(\rho(V_i))^\omega + (\rho(V_i))^{\omega+1}) \in S$ (here, we use ``$\omega$'' to denote a multiplicative idempotent power for both $R$ and $Q$). Since this holds for every $i \leq \ell$ and $E$ is idempotent, it then follows from closure under multiplication for $S$ that,
\[
(E,\prod_{1 \leq i \leq \ell} \left((\rho(V_i))^\omega + (\rho(V_i))^{\omega+1}\right)) \in S
\]
For $i \leq \ell$, let $r_i = (\rho(V_1))^\omega \cdots (\rho(V_{i-1}))^\omega (\rho(V_i))^{\omega+1} (\rho(V_{i+1}))^\omega \cdots (\rho(V_\ell))^\omega \in R$. One may verify that $r_1 + \cdots + r_\ell \leq \prod_{1 \leq i \leq \ell} \left((\rho(V_i))^\omega + (\rho(V_i))^{\omega+1}\right))$ by distributing the multiplication in the right side of this inequality. Since $S$ is closed under downset (it is \sfp{\Cs}-saturated), this yields $(E,r_1 + \cdots + r_\ell) \in S$. By definition of $T$, we get $(\{E\},r_1 + \cdots + r_\ell) \in T$.

For $i \leq \ell$, let $q_i = (\gamma(V_1))^\omega \cdots (\gamma(V_{i-1}))^\omega (\gamma(V_i))^{\omega+1} (\gamma(V_{i+1}))^\omega \cdots (\gamma(V_\ell))^\omega \in Q$. Since $\{E\}$ is idempotent for both addition and multiplication in $2^{\sclac}$, and $\gamma(V_i) = (\{E\},\rho(V_i))$, one may verify that $q_1 + \cdots + q_\ell  = (\{E\},r_1 + \cdots r_\ell) \in T$. Since $t \in T$ by definition and $T$ is closed under multiplication, we obtain $t(q_1 + \cdots + q_\ell) \in T$. Moreover, one may verify from Lemma~\ref{lem:sfclos:basecase} that for every $i \leq \ell$, we have $tq_i= t\gamma(V_i)$. Altogether, this yields that $t\gamma(K) = t(q_1 + \cdots + q_\ell) \in T$ which concludes the idempotent case.

\medskip
\noindent
{\bf General case.} It remains to handle the case when $K = P \cap C^*$ for an arbitrary $C \in \sclac$. By Fact~\ref{fct:sfclos:images}, $K$ is the union of all $\Hb$-products $V$ such that $\ctype{V} = C$. Since $K$ is non empty, there exists a least one such $V$. We fix it for the proof. Since $(Q,\cdot)$ is a finite monoid, there exists a number $p \geq 1$ such that $\gamma(V^p)$ is a multiplicative idempotent of $Q$. Since $\ctype{V}= C$, we have $\gamma(V) = (\{C\},\rho(V))$. Hence, $\gamma(V^p) = (\{C\}^p,\rho(V^p))$ and the multiplication of $p$ copies of $C$ with ``$\cmult$'' is an idempotent $E \in \sclac$ (in particular, we have $C^p \subseteq E$).

We know that $P^* \cap E \neq \emptyset$ (it includes all words in $V^p$). Hence, since we already handled the idempotent case, we know that,
\[
(\{E\},\rho(P^* \cap E)) = \tau(P^* \cap E) \in T 
\]
This yields $(E,\rho(P^* \cap E)) \in S$ by definition of $T$. We have $P^*V^{p-1} \subseteq P^*$ since $V \subseteq P^*$ ($V$ is the concatenation of languages in \Hb which is a partition of $P$). Moreover, $C V^{p-1} \subseteq E$ since $V \subseteq C$ and $C^p \subseteq E$. Altogether, we obtain $(P^* \cap C)V^{p-1} \subseteq P^* \cap E$. Therefore, $\rho((P^* \cap C)V^{p-1}) \leq \rho(P^* \cap E)$ and since $S$ is closed under downset (it is \sfp{\Cs}-saturated), this yields $(E,\rho((P^* \cap C)V^{p-1})) \in S$. By definition of $T$, this implies that,
\[
\gamma((P^* \cap C)V^{p-1}) = (\{E\},\rho((P^* \cap C)V^{p-1})) \in T
\]
Moreover, we have $t,\gamma(V) \in T$. Since $T$ is closed under multiplication, we get that $t\gamma(K)\gamma(V^{p}) = t\gamma(P^* \cap C)\gamma(V^{p}) \in T$. We already established that $\gamma(K) \in \csats{\Hb}$. Moreover, $\gamma(V^{p})$ is idempotent by definition. Hence, Lemma~\ref{lem:sfclos:basecase} yields that $t\gamma(K)\gamma(V^{p}) = t\gamma(K)$. Altogether, we obtain $t\gamma(K) \in T$, concluding the proof.

\subsubsection*{Inductive step: $t$ is not $\Hb$-stable}

The hypothesis that $t$ is not $\Hb$-stable yields some language $H \in \Hb$ such that the following \textbf{strict} inclusion holds:
\begin{equation} \label{eq:sfclos:coverind}
t \cdot \csats{\Hb} \cdot \gamma(H) \subsetneq t \cdot \csats{\Hb}
\end{equation}
We fix this language $H \in \Hb$ for the remainder of the proof. Using induction on our second parameter in Proposition~\ref{prop:sfclos:pumping}, we prove the following lemma.

\begin{lemma} \label{lem:sfclos:covalphind}
	There exists a \bsdp{\Cs}-partition \Ub of $(P \setminus H)^*$ such that for every $U \in \Ub$, we have $\gamma(U) \in \csats{\Hb}$ and $\gamma(U) \in T$.
\end{lemma}

\begin{proof}
	Fact~\ref{fct:sfclos:sprefsub} implies that $P \setminus H$ remains a prefix code with bounded synchronization delay since $P$ was one. We want to apply induction in Proposition~\ref{prop:sfclos:pumping} for the case when $P$ has been replaced by $P \setminus H$. This requires defining a \bsdp{\Cs}-partition \Gb of $P \setminus H$ and verifying that our inductions parameters have decreased. We let $\Gb = \Hb \setminus \{H\}$ which is a \bsdp{\Cs}-partition of $P \setminus H$ since \Hb was a \bsdp{\Cs}-partition of $P$. Moreover, $\tau(G) \in T$ for every $G \in \Gb$ by hypothesis on \Hb. Finally, it is immediate that $\csatp{\Gb} \subseteq \csatp{\Hb}$ (our first induction has not increased) and $\Gb \subsetneq \Hb$ (our second parameter has decreased). Hence, we may apply Proposition~\ref{prop:sfclos:pumping} in the case when $P,\Hb$ and $t \in T$ have been replaced by $P \setminus H,\Gb$ and $1_Q \in T$. This yields a \bsdp{\Cs}-partition \Ub of $(P \setminus H)^*$ such that for every $U \in \Ub$, we have $\gamma(U) \in \csats{\Gb} \subseteq \csats{\Hb}$ and $\gamma(U) \in T$.
\end{proof}

We fix the \bsdp{\Cs}-partition \Ub of $(P \setminus H)^*$ given by Lemma~\ref{lem:sfclos:covalphind} for the remainder of the proof. We distinguish two sub-cases.

Since $H \in \Hb$, one may verify from the definitions of \csatp{\Hb} and \csats{\Hb} that $\csats{\Hb} \cdot \gamma(H) \subseteq \csatp{\Hb}$. We consider two sub-cases depending on whether this inclusion is strict or not.

\subsubsection*{Sub-case~1: we have the equality $\csats{\Hb} \cdot \gamma(H) = \csatp{\Hb}$}

Recall that we have to exhibit a \bsdp{\Cs}-partition \Kb of $P^*$ which satisfies~\eqref{eq:sfclos:covergoal}. In this case, the argument is based on the following lemma which is proved using our hypotheses and induction on our third parameter (i.e. the size of $t \cdot \csats{\Hb}$).

\begin{lemma} \label{lem:sfclos:covbufind}
	For every $U \in \Ub$, there exists a \bsdp{\Cs}-partition $\Wb_{U}$ of $P^*$ such that for every $W \in \Wb_U$, $\gamma(W) \in \csats{\Hb}$ and $t\gamma(UHW) \in T$.
\end{lemma}

\begin{proof}
	We fix $U \in \Ub$ for the proof. By definition of \Ub in Lemma~\ref{lem:sfclos:covalphind}, $\gamma(U) \in \csats{\Hb}$ and $\gamma(U) \in T$. Hence, since $t,\gamma(H)\in T$ by hypothesis and $T$ is closed under multiplication, we have $t\gamma(U)\gamma(H) \in T$. Moreover, it is clear that $t\gamma(U)\gamma(H) \cdot \csats{\Hb} \subseteq t \cdot \csats{\Hb} \cdot \gamma(H) \cdot \csats{\Hb} \subseteq t\cdot \csatp{\Hb}$. Combined with our hypothesis in Sub-case~1 (i.e. $\csats{\Hb} \cdot \gamma(H) = \csatp{\Hb}$), this yields $t\gamma(U)\gamma(H) \cdot \csats{\Hb} \subseteq t \cdot \csatp{\Hb}$. We may then use~\eqref{eq:sfclos:coverind} ($t \cdot \csats{\Hb} \cdot \gamma(H) \subsetneq t \cdot \csats{\Hb}$) to get the \textbf{strict} inclusion $t\gamma(U)\gamma(H) \cdot \csats{\Hb} \subseteq t \cdot \csats{\Hb}$. Consequently, we may use induction on our third parameter (i.e. the size of $t \cdot \csats{\Hb}$) to apply Proposition~\ref{prop:sfclos:pumping} in the case when $t \in T$ has been replaced by $t\gamma(U)\gamma(H) \in T$. Note that here, our first two parameters have not increased (they only depend on \Hb which remains unchanged). This yields the desired \bsdp{\Cs}-partition $\Wb_{U}$ of $P^*$.
\end{proof}

It remains to use Lemma~\ref{lem:sfclos:covbufind} to conclude the proof of Sub-case~1. We build our \bsdp{\Cs}-partition \Kb of $P^*$ as follows,
\[
\Kb = \Ub \cup \bigcup_{U \in \Ub} \{UHW \mid W \in \Wb_U\}
\]
It remains to show that \Kb is a \bsdp{\Cs}-partition of $P^*$ which satisfies~\eqref{eq:sfclos:covergoal}. First, observe that the languages $K \in \Kb$ belong to \bsdp{\Cs}. This is immediate by definition of \Ub is $K \in \Ub$. Otherwise, $K = UHW$ with $U \in \Ub$ and $W \in \Wb_U$ and $U,H,W \in \bsdp{\Cs}$. Moreover, one may verify that the concatenation $UHW$ is \emph{unambiguous} since $P$ is a prefix code, $U \subseteq (P \setminus H)^*$ and $W \subseteq H^*$. Altogether, it follows that $K \in \bsdp{\Cs}$.

That \Kb is a partition of $P^*$ is also simple to verify from the definition since $P$ is a prefix code, $H \subseteq P$ and $\Ub,\Wb_{U}$ are partitions of $(P \setminus H)^*$ and $P^*$ respectively. Every word $w \in P^*$ admits a \emph{unique} decomposition $w = w_1 \cdots w_n$ with $w_1,\dots,w_n \in P$. We partition $P^*$ by looking at the leftmost factor belonging to $H$ (when it exists).

It remains to prove that~\eqref{eq:sfclos:covergoal} holds. Consider $K \in \Kb$, we show that $\gamma(K) \in \csats{\Hb}$ and $t\gamma(K) \in T$. If $K \in \Ub$, this is immediate by definition of \Ub in Lemma~\ref{lem:sfclos:covalphind}. Otherwise, $K = UHW$ with $U \in \Ub$ and $W \in \Wb_U$. By definition of \Ub and $\Wb_U$, we have $\gamma(U),\gamma(W)\in\csats{\Hb}$. Hence, $\gamma(UHW) \in \csats{\Hb}$. Moreover, we have $t\gamma(UHW) \in T$ by definition of $\Wb_U$ in Lemma~\ref{lem:sfclos:covbufind}. This concludes the first sub-case.

\subsubsection*{Sub-case~2: we have the strict inclusion $\csats{\Hb} \cdot \gamma(H) \subsetneq \csatp{\Hb}$}

Recall that our objective is to construct a \bsdp{\Cs}-partition \Kb of $P^*$ which satisfies~\eqref{eq:sfclos:covergoal}. We begin by giving a brief overview of the construction. Consider a word $w \in P^*$. Since $P$ is a prefix code, $w$ admits a unique decomposition as a concatenation of factors in $P$. We may look at the rightmost factor in $H \subseteq P$ to uniquely decompose $w$ in two parts (each of them possibly empty): a prefix in $((P \setminus H)^*H)^*$ and a suffix in $(P \setminus H)^*$. We use induction to construct \bsdp{\Cs}-partitions of the sets of possible prefixes and suffixes. Then, we combine them to construct a \bsdp{\Cs}-partition of the whole set $P^*$. Actually, we already constructed a suitable \bsdp{\Cs}-partition of the possible suffixes in $(P \setminus H)^*$: \Ub (see Lemma~\ref{lem:sfclos:covalphind}). Hence, it remains to partition the prefixes. We do so this in the following lemma which is proved using the hypothesis of Sub-case~2 and induction on our first parameter.

\begin{lemma} \label{lem:sfclos:covindratm}
	There exists a \bsdp{\Cs}-partition \Vb of $((P \setminus H)^*H)^*$ such that for every $V \in \Vb$, $\gamma(V) \in \csats{\Hb}$ and $\gamma(V) \in T$.
\end{lemma}

\begin{proof}
	Let $L = (P\setminus H)^*H$. Fact~\ref{fct:sfclos:sprefnest} implies that $L$ remains a prefix code with bounded synchronization delay since $P$ was one. We want to apply  induction in Proposition~\ref{prop:sfclos:pumping} for the case when $P$ has been replaced by $L$. Doing so requires building an appropriate \bsdp{\Cs}-partition of $L$ and proving that one of our induction parameters has decreased.
	
	Let $\Fb = \{UH \mid U \in \Ub\}$. Since \Ub is a partition of $(P\setminus H)^*$ and $P$ is a prefix code, one may verify that \Fb is a partition of $L = (P\setminus H)^*H$. Moreover, it is a \bsdp{\Cs}-partition of $L$: every $F \in \Fb$ is the unambiguous concatenation of two languages in \bsdp{\Cs}. Moreover, given $F \in \Fb$, we have $F = UH$ for $U \in \Ub$ which means that $\gamma(F) = \gamma(U)\gamma(H) \in T$ since $T$ is closed under multiplication. It remains to show that our induction parameters have decreased. Since $\Fb = \{UH \mid U \in \Ub\}$ and $\gamma(U) \in \csats{\Hb}$ for every $U \in \Ub$ (by definition of \Ub in Lemma~\ref{lem:sfclos:covalphind}), one may verify that $\csatp{\Fb} \subseteq \csats{\Hb} \cdot \gamma(H)$. Hence, since $\csats{\Hb} \cdot \gamma(H) \subsetneq \csatp{\Hb}$ by hypothesis in Sub-case~2, we have $\csatp{\Fb} \subsetneq \csatp{\Hb}$. Our first induction parameter has decreased. Altogether, it follows that we may apply Proposition~\ref{prop:sfclos:pumping} in the case when $P,\Hb$ and $t \in T$ have been replaced by $L,\Fb$ and $1_Q \in T$. This yields a \bsdp{\Cs}-partition \Vb of $L^* = ((P \setminus H)^*H)^*$ such that for every $V \in \Vb$, $V \in \csats{\Fb}$ and $\gamma(V) \in T$. Finally, it is clear by definition that $\csats{\Fb} \subseteq \csats{\Hb}$. Hence, the lemma follows.
\end{proof}

We are ready to construct the \bsdp{\Cs}-partition \Kb of $P^*$ and conclude the main argument. We let $\Kb = \{VU \mid V \in \Vb \text{ and } U \in \Ub\}$. It is immediate by definition that \Kb is a partition of $P^*$ since $P$ is a prefix code and $\Vb,\Ub$ are partitions of $((P \setminus H)^*H)^*$ and $(P \setminus H)^*$ respectively (see the above discussion). Additionally, it is immediate by definition that \Kb is actually a $\bsdp{\Cs}$-partition of $P^*$ (it only contains unambiguous concatenations of languages in \bsdp{\Cs}). It remains to prove that \Kb satisfies~\eqref{eq:sfclos:covergoal}: for every $K \in \Kb$, we have $\gamma(K) \in \csats{\Hb}$ and $t\gamma(K) \in T$. We have $K = VU$ with $V \in \Vb$ and $U \in \Ub$. By definition of \Ub and \Vb, we have $\gamma(U),\gamma(V) \in \csats{\Hb}$ and $\gamma(U),\gamma(V) \in T$. Moreover, $t \in T$ by hypothesis. Therefore, since both $\csats{\Hb}$ and $T$ are closed under multiplication, it follows that $\gamma(K) \in \csats{\Hb}$ and $t\gamma(K) \in T$. This concludes the proof of Proposition~\ref{prop:sfclos:pumping}.
 
\section{Appendix to Section~\ref{sec:units}}
\label{app:appcovg}
This appendix presents the missing proofs in Section~\ref{sec:units}. Similarly to what we did in Appendix~\ref{app:sfc}, we actually present a self-contained, full version of Section~\ref{sec:units} which can be read independently. Recall that the purpose of Section~\ref{sec:units} is to prove Theorem~\ref{thm:sfclos:gmain}: when \Cs is a \vari of group languages with decidable separation, \sfp{\Cs}-covering is decidable. 
As in Section~\ref{sec:finite}, we rely on Proposition~\ref{prop:breduc}: we present an algorithm which computes \sfcopti from an input \nice \mratm $\rho$. In this case as well, we do not work with \sfcopti itself. Instead, we consider another set carrying  more information. Defining this second object involves introducing a few additional concepts. We first present them and then turn to the algorithm. For details, see~\cite{concagroup}.

\subsection{Preliminary definitions}

\noindent
\textbf{Optimal \idens.} In this case, handling \sfp{\Cs} involves considering \Cs-optimal covers of $\{\veps\}$. Since $\{\veps\}$ is a singleton, there always exists such a cover consisting of a single language, which leads to the following definition.

Let \Cs be a Boolean algebra (we shall use the case when \Cs contains only group languages but this is not required for the definitions) and $\tau: 2^{A^*} \to Q$ be a \ratm. A \emph{\Cs-optimal \iden for $\tau$} is a language $L \in \Cs$ such that $\veps \in L$ and  $\tau(L) \leq \tau(L')$ for every $L' \in \Cs$ satisfying $\veps \in L'$. As expected, there always exists a \Cs-optimal \iden for any \ratm $\tau$.

\begin{lemma} \label{lem:aliden}
	Let $\tau: 2^{A^*} \to Q$ be a \ratm and \Cs be a Boolean algebra. There exists a \Cs-optimal \iden for $\tau$.
\end{lemma}

\begin{proof}
  Let $U \subseteq Q$ be the set of all elements $q \in Q$ such that $q = \tau(K)$ for some $K \in \Cs$ containing $\veps$. Clearly, $\tau(A^*) \in U$ which means that $U$ is non-empty since $A^* \in \Cs$ (\Cs is a Boolean algebra). For every $q \in U$, we fix an arbitrary language $K_q \in \Cs$ such that $\veps \in K_q$ and $q = \tau(K_q)$ ($K_q$ exists by definition of $U$). Finally, we let,
	\[
	K = \bigcap_{q \in U} K_q
	\]
	Since \Cs is a Boolean algebra, we have $K \in \Cs$. Moreover, $\veps \in K$ by definition. Since $K \subseteq K_q$ for all $q \in U$, it follows that $\tau(K) \leq q$ for every $q \in U$. By definition of $U$, this implies that $\tau(K) \leq \tau(K')$ for every $K' \in \Cs$ such that $\veps \in K'$. Hence, $K$ is a \Cs-optimal \iden for $\tau$.
\end{proof}

Moreover, all \Cs-optimal \idens for $\tau$ have the same image under $\tau$. We write it $\ioptic{\tau} \in Q$: $\ioptic{\tau} = \tau(L)$ for every \Cs-optimal \iden $L$ for $\tau$. It turns out that when $\tau$ is \nice and \tame, computing \iopti{\Cs}{\tau} from $\tau$ boils down to \Cs-separation. This is important: this is exactly how our algorithm for \sfp{\Cs}-covering depends on \Cs-separation.

\adjustc{lem:sepepswit}
\begin{lemma}
	Let $\tau: 2^{A^*} \to Q$ be a \nice \ratm and \Cs a Boolean algebra. Then, $\iopti{\Cs}{\tau}$ is the sum of all $q \in Q$ such that $\{\veps\}$ is not \Cs-separable from $\tau_*\inv(q)$.
\end{lemma}
\restorec

\begin{proof}
  We let $U \subseteq Q$ be the set of all $q \in Q$ such that $\{\veps\}$ is not \Cs-separable from $\tau_*\inv(q)$. Moreover, we let $r = \sum_{q \in U} q$. We show that $r=\iopti{\Cs}{\tau}$. First, we prove that $r\leq\iopti{\Cs}{\tau}$. By definition, this amounts to proving that given $q \in U$, we have $q \leq \iopti{\Cs}{\tau}$. By definition, $\iopti{\Cs}{\tau} = \tau(K)$ for some $K \in \Cs$ such that $\veps \in K$. By definition of $U$, $\{\veps\}$ is not \Cs-separable from $\tau_*\inv(q)$ which implies that $K \cap\tau_*\inv(q) \neq \emptyset$. Let $w \in K \cap\tau_*\inv(q)$. By definition $w \in K$ and $\tau(w) = q$. This implies that $q \leq \tau(K) = \iopti{\Cs}{\tau}$.
	
	Conversely, we show that $\iopti{\Cs}{\tau} \leq r$. By definition, for every $q \in Q \setminus U$, $\{\veps\}$ is \Cs-separable from $\tau_*\inv(q)$. We  fix a language $H_q \in \Cs$ as a separator. We now let,
	\[
	H = \bigcap_{q \in Q \setminus U} H_q \in \Cs
	\]
	By definition, $H$ separates $\{\veps\}$ from $\tau_*\inv(q)$ for every $q \in Q \setminus U$. In particular, $\veps \in H$ which implies that $\iopti{\Cs}{\tau} \leq \tau(H)$. Moreover, since $\tau$ is \nice, we have $w_1,\dots,w_n \in H$ such that $\tau(H) = \tau(w_1) + \cdots + \tau(w_n)$. Finally, since $H \cap \tau_*\inv(q) = \emptyset$ for every $q \in Q \setminus U$, we know that $\tau(w_1),\dots,\tau(w_n) \in U$. This implies that $\tau(H) \leq r$. Altogether, we get that $\iopti{\Cs}{\tau} \leq r$.
\end{proof}

\noindent
\textbf{Nested \ratms.} We want an algorithm which computes \sfcopti from an input \nice \mratm $\rho$ for a fixed \vari of group languages \Cs. Yet, we shall not use optimal \idens with this input \ratm $\rho$. Instead, we consider an auxiliary \ratm built from $\rho$ (the definition is taken from~\cite{pzbpolc}).

Consider a Boolean algebra \Ds (we shall use the case $\Ds = \sfp{\Cs}$) and a \ratm $\rho: 2 ^{A^*} \to R$. We build a new map $\bratauxd: 2^{A^*} \to 2^R$ whose rating set is $(2^R,\cup)$. For every $K \subseteq A^*$, we define $\bratauxd(K) = \opti{\Ds}{K,\rho} \in 2^R$. It follows from Fact~\ref{fct:lunion} that this is indeed a \ratm (on the other hand \bratauxd need not be \nice, see~\cite{pzbpolc} for a counterexample). More importantly, \bratauxd need not be \tame, even if this is the case for $\rho: 2 ^{A^*} \to R$. However, it turns out that we are able to cope with this last issue.

Clearly, when $\rho: 2 ^{A^*} \to R$ is a \mratm, the set $2^R$ is an idempotent semiring (addition is union and the multiplication is lifted from the one of $R$). It turns out that when \Ds is \vari closed under concatenation (such as when $\Ds = \sfp{\Cs}$), $\bratauxd: 2^{A^*} \to 2^R$ behaves almost as a \mratm for this semiring structure on $2^R$. Indeed, we have the following statement taken from~\cite{pzbpolc} (see Lemma~6.7). For $S \subseteq R$, we write $\dclosr S \subseteq R$ for the set $\dclosr S = \{q \mid \text{$q \leq r$ for some $r \in S$}\}$.

\begin{lemma} \label{lem:quasitame}
	Consider a \vari \Cs and a \mratm $\rho: 2 ^{A^*} \to R$. Then, for every $H,K \subseteq A^*$, we have $\dclosr \bratauxsfc(K) = \bratauxsfc(K)$ and $\dclosr (\bratauxsfc(H) \cdot \bratauxsfc(K)) = \bratauxsfc(HK)$.
\end{lemma}

We may now explain which set is computed by our algorithm instead of \sfcopti. Consider a \nice \mratm $\rho: 2^{A^*} \to R$. Since $\bratauxsfc: 2^{A^*} \to 2^R$ is a \ratm, we may consider the element $\ioptic{\bratauxsfc} \in 2^R$. By definition, $\ioptic{\bratauxsfc} = \bratauxsfc(L)$ where $L$ is a \Cs-optimal \iden for $\bratauxsfc$. Therefore, \ioptic{\bratauxsfc} is a subset of $\bratauxsfc(A^*) = \opti{\sfp{\Cs}}{A^*,\rho} = \sfcopti$. When \Cs is a \vari of group languages, one may compute the whole set \sfcopti from this subset.

\adjustc{prop:utooptibool}
\begin{proposition}
	Let \Cs be a \vari of group languages and $\rho: 2^{A^*} \to R$ a \nice \mratm. Then, \sfcopti is the least subset of $R$ containing \ioptic{\bratauxsfc} and satisfying the three following properties:
	\begin{itemize}
		\item {\bf Trivial elements.} For every $w \in A$, $\rho(w) \in \sfcopti$.
		\item {\bf Downset.} For every $r \in \sfcopti$ and $q \leq r$, we have $q \in \sfcopti$.
		\item {\bf Multiplication.} For every $q,r \in \sfcopti$, we have $qr \in \sfcopti$.
	\end{itemize}
\end{proposition}
\restorec

\begin{proof}
  We let $S \subseteq R$ be the least subset of $R$ containing \ioptic{\bratauxsfc} and satisfying the three properties in the proposition. We need to prove that $S = \sfcopti$. The inclusion $S \subseteq \sfcopti$ is immediate. We already observed before the statement that $\ioptic{\bratauxsfc} \subseteq \sfcopti$. The other properties are generic ones of optimal \imprints when the investigated class is a \vari (see Lemma~9.5 in~\cite{pzcovering2}).
	
  We concentrate on the converse inclusion: $\sfcopti \subseteq S$. For the proof, we let $L \subseteq A^*$ be a \Cs-optimal \iden for \bratauxsfc. Recall that by definition, this means that $L \in \Cs$, $\veps\in L$ and $\ioptic{\bratauxsfc} = \bratauxsfc(L)$. Consider $r \in \sfcopti$, we show that $r \in S$. The argument is based on the following lemma which is where we use the hypothesis that \Cs is made of \emph{group languages}.
		
	\begin{lemma} \label{lem:nest:grouplem}
		There exist $\ell \in \nat$ and $a_1,\dots,a_\ell \in A$ such that $r \in \bratauxsfc(La_1L \cdots a_\ell L)$.
	\end{lemma}
	
	\begin{proof}
		Since $L \in \Cs$, it is a group language by hypothesis on \Cs: it is recognized by a finite group. Therefore, since $\veps \in L$, one may use a pumping argument to show that $A^*$ is a finite union of languages having the form $La_1L \cdots a_\ell L$ for $a_1,\dots,a_\ell \in A$. Moreover, $\sfcopti = \opti{\sfp{\Cs}}{A^*,\rho} = \bratauxsfc(A^*)$ by definition. Therefore, since \bratauxsfc is a \ratm (whose rating set is $(2^R,\cup)$, it follows that \sfcopti is a finite union of sets $\bratauxsfc(La_1L \cdots a_\ell L)$ for $a_1,\dots,a_\ell \in A$. Since $r \in \sfcopti$, the result follows.
	\end{proof}
	
	In view of Lemma~\ref{lem:quasitame},  the hypothesis that $r \in \bratauxsfc(La_1L \cdots a_\ell L)$ given by Lemma~\ref{lem:nest:grouplem} implies that,
	\[
	r \in \dclosr(\bratauxsfc(L) \cdot \bratauxsfc(a_1) \cdot \bratauxsfc(L) \cdots \bratauxsfc(a_\ell) \cdot \bratauxsfc(L))
	\]
	By definition of $L$, we have $\bratauxsfc(L) = \ioptic{\bratauxsfc}$. Moreover, we have the following fact.
	
	\begin{fact} \label{fct:theletters}
		For every $a \in A$, we have $\bratauxsfc(a) = \dclosr \{\rho(a)\}$.
	\end{fact}
	
	\begin{proof}
		By definition, $\bratauxsfc(a) = \opti{\sfp{\Cs}}{\{a\},\rho}$. Since $\{a\} \in \sfp{\Cs}$ by definition, it is clear that $\{\{a\}\}$ is an optimal \sfp{\Cs}-cover of $\{a\}$. Thus, $\opti{\sfp{\Cs}}{\{a\},\rho} = \prin{\rho}{\{a\}} = \dclosr \{\rho(a)\}$ and we are finished.
	\end{proof}
	
	Altogether, we obtain that,
	\[
	r \in \dclosr(\ioptic{\bratauxsfc} \cdot \{\rho(a_1)\} \cdot \ioptic{\bratauxsfc} \cdots \{\rho(a_\ell)\} \cdot \ioptic{\bratauxsfc})
	\]
	It is now immediate from the definition of $S$ that $r \in S$ which concludes the proof for the inclusion $\sfcopti \subseteq S$.
\end{proof}

\subsection{Algorithm}

We may now present our algorithm for computing \sfcopti. We fix a \vari of group languages \Cs for the presentation. As expected, the main procedure computes \ioptic{\bratauxsfc} (see Proposition~\ref{prop:utooptibool}). In this case as well, this procedure is obtained from a characterization theorem.

Consider a \nice \mratm $\rho: 2^{A^*} \to R$. We define the \sfp{\Cs}-complete subsets of $R$ for $\rho$. The definition depends on auxiliary \nice \mratms. We first present them. Clearly, $2^R$ is an idempotent semiring (addition is union and the multiplication is lifted from the one of $R$). For every $S \subseteq R$, we use it as the rating set of a \nice \mratm $\sfrats: 2^{A^*} \to 2^{R}$. Since we are defining a \emph{\nice} \mratm, it suffices to specify the evaluation of letters. For $a \in A$, we let $\sfrats(a) = S \cdot \{\rho(a)\} \cdot S \in 2^R$. Observe that by definition, we have $\ioptic{\sfrats} \subseteq R$.

We are ready to define the \sfp{\Cs}-complete subsets of $R$. Consider $S \subseteq R$. We say that $S$ is \sfp{\Cs}-complete for $\rho$ when the following conditions are satisfied:
\begin{enumerate}
	\item\label{clos:sfc:downg} {\bf Downset.} For every $r \in S$ and $q \leq r$, we have $q \in S$.
	\item\label{clos:sfc:multg} {\bf Multiplication.} For every $q,r \in S$, we have $qr \in S$.
	\item\label{clos:sfc:operg} {\bf \Cs-operation.} We have $\ioptic{\sfrats} \subseteq S$.
	\item\label{clos:sfc:closg} {\bf \sfp{\Cs}-closure.} For every $r \in S$, we have $r^\omega + r^{\omega+1} \in S$.
\end{enumerate}

\adjustc{rem:sfclos:compne}
\begin{remark}
	The definition of \sfp{\Cs}-complete subsets does not explicitly require that they contain some trivial elements. Yet, this is implied by \Cs-operation. Indeed, if $S \subseteq R$ is \sfp{\Cs}-complete, then $\sfrats(\veps) = \{1_R\}$ (this is the multiplicative neutral element of $2^{R}$). This implies that $1_R \in \ioptic{\sfrats}$ and we obtain from \Cs-operation that $1_R \in S$.
\end{remark}
\restorec

\adjustc{thm:sfclos:cargroup}
\begin{theorem}[\sfp{\Cs}-optimal \imprints (\Cs made of group languages)]
Let $\rho: 2^{A^*} \to R$ be a \nice \mratm. Then, \ioptic{\bratauxsfc} is the least \sfp{\Cs}-complete subset of $R$.	
\end{theorem}
\restorec

When \Cs-separation is decidable, Theorem~\ref{thm:sfclos:cargroup} yields a least fixpoint procedure for computing \ioptic{\bratauxsfc} from a \nice \mratm $\rho: 2^{A^*} \to R$. The computation starts from the empty set and saturates it with the four operations in the definition of \sfp{\Cs}-complete subsets. It is clear that we may implement downset, multiplication and \sfp{\Cs}-closure. Moreover, we may implement \Cs-operation as this boils down to \Cs-separation by Lemma~\ref{lem:sepepswit}. Eventually, the computation reaches a fixpoint and it is straightforward to verify that this set is the least \sfp{\Cs}-complete subset of $R$, \emph{i.e.}, \ioptic{\bratauxsfc} by Theorem~\ref{thm:sfclos:cargroup}.

\smallskip

By Proposition~\ref{prop:utooptibool}, we may compute \sfcopti from \ioptic{\bratauxsfc}. Altogether, this yields the decidability of \sfp{\Cs}-covering by Proposition~\ref{prop:breduc}. Hence, Theorem~\ref{thm:sfclos:gmain} is proved.

\subsection{Proof}

We now concentrate on proving Theorem~\ref{thm:sfclos:cargroup}. We fix a \nice \mratm $\rho: 2^{A^*} \to R$ for the proof. We need to show that \ioptic{\bratauxsfc} is the least \sfp{\Cs}-complete subset of $R$. The proof involves soundness and completeness directions. In both cases, we apply Theorem~\ref{thm:sfclos:carfinite} as a sub-result.

\subsubsection*{Soundness}

We prove that \ioptic{\bratauxsfc} is \sfp{\Cs}-complete. We start with a preliminary simple fact that will be useful.

\begin{fact} \label{fct:idengroup}
	There exists a finite group $G$ and a morphism $\alpha: A^* \to G$ such that the language $L = \alpha\inv(1_G)$ is a \Cs-optimal \iden for \bratauxsfc.
\end{fact}

\begin{proof}
  We let $H$ be a \Cs-optimal \iden for \bratauxsfc: we have $H \in \Cs$, $\veps \in H$ and $\bratauxsfc(H) = \ioptic{\bratauxsfc}$. Let $\alpha: A^* \to G$ be the syntactic morphism of $H$. Since \Cs is a \vari of group languages and $H \in \Cs$, it is standard that $G$ is a finite group and that every language recognized by $\alpha$ belongs to \Cs (see~\cite{pingoodref} for example). In particular $L = \alpha\inv(1_G) \in \Cs$. Moreover, since $\veps \in H$ and $H$ is recognized by $\alpha$, we have $L \subseteq H$. Since $H$ is a \Cs-optimal \iden for \bratauxsfc, $L$ must be one as well.
\end{proof}

We fix the morphism $\alpha: A^* \to G$ and the \Cs-optimal \iden for \bratauxsfc $L = \alpha\inv(1_G)$ for the proof. We need to show that $\bratauxsfc(L) = \ioptic{\bratauxsfc}$ is \sfp{\Cs}-complete. There are four properties to verify.

\medskip
\noindent
{\bf Downset (Condition~\ref{clos:sfc:downg}).} By definition of $L$, we have $\ioptic{\bratauxsfc} = \bratauxsfc(L)$ and it is immediate from Lemma~\ref{lem:quasitame} that $\bratauxsfc(L) = \dclosr \bratauxsfc(L)$.

\medskip
\noindent
{\bf Multiplication (Condition~\ref{clos:sfc:multg}).} Let $q,r \in \ioptic{\bratauxsfc}$. We show that $qr \in \ioptic{\bratauxsfc}$. By hypothesis, we have $q,r \in \bratauxsfc(L)$. In particular, this implies that,
\[
qr \in \dclosr (\bratauxsfc(L) \cdot \bratauxsfc(L))
\]
Hence, Lemma~\ref{lem:quasitame} yields that $qr \in \bratauxsfc(LL)$. Finally, since $L = \alpha^{-1}(1_G)$, we have $LL = L$. Thus, $qr \in \bratauxsfc(L) = \ioptic{\bratauxsfc}$.

\medskip
\noindent
{\bf \Cs-operation (Condition~\ref{clos:sfc:operg}).} For the sake of avoiding clutter, we write $S = \ioptic{\bratauxsfc}$. We need to show that $\ioptic{\sfrats} \subseteq S$. Hence, we consider $r \in \ioptic{\sfrats}$ and show that $r \in S = \bratauxsfc(L)$. Since $L \in \Cs$ and $\veps \in L$, the hypothesis that $r \in \ioptic{\sfrats}$ yields $r \in \sfrats(L)$. Hence, we get $w \in L$ such that $r \in \sfrats(w)$. There are now two cases depending on whether $w = \veps$ or $w \in A^+$.

Assume first that $w = \veps$. In that case, we have $\sfrats(w) = \{1_R\}$ (this is the multiplicative neutral element of $2^R$). Thus, $r = 1_R$. Since $\veps \in L$, it is clear that $1_R \in \opti{\sfp{\Cs}}{L,\rho} = \bratauxsfc(L) = S$ which concludes this case.

We now assume that $w \in A^+$. In that case, there exist $n \geq 1$ and $a_1,\dots,a_n \in A$ such that $w = a_1 \cdots a_n$. By definition $\sfrats(a) =  S \cdot \{\rho(a_1)\} \cdot S$ for every $a \in A$. Moreover, since we already established that $S = \ioptic{\bratauxsfc}$ is closed under multiplication, we have $S \cdot S \subseteq S$. Hence, it follows that 
\[
\sfrats(w) \subseteq S \cdot \{\rho(a_1)\} \cdot S \cdots \{\rho(a_n)\} \cdot S
\]
We get $r \in S \cdot \{\rho(a_1)\} \cdot S \cdots \{\rho(a_n)\} \cdot S$. By definition, we have $S = \bratauxsfc(L)$ and it is clear that for every $a \in A$, we have $\rho(a) \in \opti{\sfp{\Cs}}{\{a\},\rho} = \bratauxsfc(a)$. Thus, we obtain that,
\[
r \in \bratauxsfc(L) \cdot \prod_{1 \leq i \leq n} \left(\bratauxsfc(a_i) \cdot \bratauxsfc(L)\right)
\]
By Lemma~\ref{lem:quasitame}, we have,
\[
\bratauxsfc(La_1L \cdots La_nL) = \dclosr \left(\bratauxsfc(L) \cdot \prod_{1 \leq i \leq n} \left(\bratauxsfc(a_i) \cdot \bratauxsfc(L)\right)\right)
\]
Consequently, $r  \in \bratauxsfc(La_1L \cdots La_nL)$. Recall that $L = \alpha\inv(1_G)$ and $a_1 \cdots a_n \in L$ (\emph{i.e.}, $\alpha(a_1 \cdots a_n) = 1_G$). Thus, $\alpha$ maps every word in $La_1L \cdots La_nL$ to $1_G$ and we obtain that $La_1L \cdots La_nL \subseteq L$. Since \bratauxsfc is a \ratm, this yields $\bratauxsfc(La_1L \cdots La_nL) \subseteq \bratauxsfc(L)$. Altogether, we obtain $r \in \bratauxsfc(L) = S$, finishing the proof.

\medskip
\noindent
{\bf \sfp{\Cs}-closure (Condition~\ref{clos:sfc:closg}).} Let $r \in \ioptic{\bratauxsfc}$. We show that $r^\omega +r^{\omega+1} \in \ioptic{\bratauxsfc}$ as well. Since $L$ is a \Cs-optimal \iden for $\bratauxsfc$, we need to show that:
\[
r^\omega +r^{\omega+1} \in \bratauxsfc(L) = \opti{\sfp{\Cs}}{L,\rho}.
\]
Let \Kb be an optimal \sfp{\Cs}-cover of $L$ for $\rho$. We need to exhibit $K \in \Kb$ such that $r^\omega +r^{\omega+1} \leq \rho(K)$. This is where we use Theorem~\ref{thm:sfclos:carfinite}. First, we use \Kb and $L$ to construct a \emph{finite} sub-class of \Cs.

\begin{fact} \label{fct:sfclos:grpaper}
	There exists a finite \vari $\Gs \subseteq \Cs$ such that $L \in \Gs$ and $K \in \sfp{\Gs}$ for every $K \in \Kb$.
\end{fact}

\begin{proof}
	By definition \Kb contains finitely many languages in \sfp{\Cs}. Hence, there exists a finite set of languages $\Hb \subseteq \Cs$ such that every language $K \in \Kb$ is built by applying Boolean operations and concatenations to languages in \Hb. Since \Cs is a \vari, it is standard that there exists a finite \vari $\Gs \subseteq \Cs$ which contains all languages in \Hb and $L$ (since $L \in \Cs$ by definition). See Lemma~17 in~\cite{pzgenconcat} for a proof. It is then immediate by definition of \Hb that $K \in \sfp{\Gs}$ for every $K \in \Kb$.
\end{proof}

Since \Gs is finite, we may consider the associated equivalence \caneg defined on $A^*$. By definition, we have $\gtype{\veps} \in \Gs \subseteq \Cs$ and $\veps\in \gtype{\veps}$. Thus, $\ioptic{\bratauxsfc}\subseteq \bratauxsfc(\gtype{\veps})$. Since $r \in \ioptic{\bratauxsfc}$ by hypothesis, this yields,
\[
r \in \bratauxsfc(\gtype{\veps}) = \opti{\sfp{\Cs}}{\gtype{\veps},\rho}.
\]
Moreover, we have $\Gs \subseteq \Cs$ which implies that $\sfp{\Gs} \subseteq \sfp{\Cs}$. Consequently, Fact~\ref{fct:lunion} yields that $r \in  \opti{\sfp{\Gs}}{\gtype{\veps},\rho}$. By definition, this property can be reformulated as follows,
\[
(\gtype{\veps},r) \in \popti{\sfp{\Gs}}{\Gs}{\rho}
\]
By Theorem~\ref{thm:sfclos:carfinite}, we know that the set 
$\popti{\sfp{\Gs}}{\Gs}{\rho}$ is $\sfp{\Gs}$-saturated for $\rho$. Hence, it satisfies $\sfp{\Gs}$-closure and since $\gtype{\veps}$ is an idempotent $\caneg$-class, we get that:
\[
(\gtype{\veps},r^\omega +r^{\omega+1}) \in \popti{\sfp{\Gs}}{\Gs}{\rho}.
\]
Therefore, we have $r^{\omega} + r^{\omega+1} \in  \opti{\sfp{\Gs}}{\gtype{\veps},\rho}$.

\medskip
Finally, we know that $L \in \Gs$ by definition of \Gs (see Fact~\ref{fct:sfclos:grpaper}). Thus, the hypothesis that $\veps \in L$ yields $\gtype{\veps} \subseteq L$. We then obtain from Fact~\ref{fct:lunion} that $\opti{\sfp{\Gs}}{\gtype{\veps},\rho} \subseteq \opti{\sfp{\Gs}}{L,\rho}$. By definition \Kb is a cover of $L$. Moreover, it is an \sfp{\Gs}-cover by definition of \Gs in Fact~\ref{fct:sfclos:grpaper}. Thus, $\opti{\sfp{\Gs}}{L,\rho} \subseteq \prin{\rho}{\Kb}$. Altogether, we obtain  that $r^{\omega} + r^{\omega+1} \in \prin{\rho}{\Kb}$ which yields $K \in \Kb$ such that $r^\omega +r^{\omega+1} \leq \rho(K)$, finishing the soundness proof.

\subsubsection*{Completeness}

We have proved that \ioptic{\bratauxsfc} is a \sfp{\Cs}-complete subset of $R$. It remains to show that it is the least such subset. Hence, we fix an arbitrary \sfp{\Cs}-complete set $S \subseteq R$ and show that $\ioptic{\bratauxsfc} \subseteq S$. By definition, we know that $\ioptic{\bratauxsfc} \subseteq \bratauxsfc(L)$ for every language $L \in \Cs$ such that $\veps \in L$. Hence, it suffices to exhibit a language $L \in \Cs$ such that $\veps \in L$ and $\bratauxsfc(L) \subseteq S$. We first choose the appropriate language $L$.

We let $H$ be a \Cs-optimal \iden for the \nice \mratm \sfrats. That is, we have $H \in \Cs$, $\veps \in H$ and $\sfrats(H) = \ioptic{\sfrats}$. Since $H \in \Cs$, one may verify that there exists a \textbf{finite} \vari $\Gs \subseteq \Cs$ such that $H \in \Gs$ (see again Lemma~17 in~\cite{pzgenconcat} for a proof). We choose $L=\gtype{\veps}$ : clearly, this language belongs to $\Gs\subseteq \Cs$ and contains \veps by definition. It now remains to show the following inclusion:
\begin{equation} \label{eq:sfclos:ginc}
  \bratauxsfc(\gtype{\veps}) \subseteq S
\end{equation}
Let us give a brief overview of the proof. It is based on Theorem~\ref{thm:sfclos:carfinite}. First, we use our set $S \subseteq R$ which is \sfp{\Cs}-complete subset (for $\rho$) to build another set $S' \subseteq (\gclac) \times R$ which is \sfp{\Gs}-saturated (for $\rho$). This is where we apply Theorem~\ref{thm:sfclos:carfinite}. It states that the \emph{least} \sfp{\Gs}-saturated set is \popti{\sfp{\Gs}}{\Gs}{\rho}. Hence, we obtain that the inclusion $\popti{\sfp{\Gs}}{\Gs}{\rho} \subseteq S'$ holds. It is then straightforward to prove~\eqref{eq:sfclos:ginc} from this inclusion.

\medskip

Let us first define the set $S' \subseteq (\gclac) \times R$. The construction is based on the set $S$ and the \nice \mratm \sfrats. We define,
\[
S' = \{(\gtype{\veps},s) \mid s \in S\} \cup \{(C,r) \mid C \in \gclac \text{ and } r \in \dclosr \sfrats(C))\}
\]
Observe that by definition, we have the following useful fact about this new set $S'$.

\begin{fact} \label{fct:sfclos:sp}
	For every $s \in R$, if $(\gtype{\veps},s) \in S'$, then $s \in S$.
\end{fact}

\begin{proof}
	Let $s \in R$ such that $(\gtype{\veps},s) \in S'$. By definition of $S'$, either $s \in S$ or $s \in \dclosr \sfrats(\gtype{\veps})$. In the former case, we are finished. Hence, we assume that $s \in \dclosr \sfrats(\gtype{\veps})$. By definition, $\gtype{\veps} \subseteq H$ (we have $H \in \Gs$ and $\veps \in H$). Since $H$ is a \Cs-optimal \iden for \sfrats, this yields $\sfrats(\gtype{\veps}) \subseteq \ioptic{\sfrats}$. hence, we have $s \in \dclosr \ioptic{\sfrats}$. Finally, since $S$ is \sfp{\Cs}-complete, we know that $\ioptic{\sfrats} \subseteq S$ and $\dclosr S = S$. Altogether, we get that $s \in S$, concluding the proof.
\end{proof}

We now turn to the technical core of the proof: our new set $S' \subseteq (\gclac) \times R$ is \sfp{\Gs}-saturated for $\rho$. We state this result in the following lemma.

\begin{lemma} \label{lem:sfclos:isat}
	The set $S' \subseteq (\gclac) \times R$ is \sfp{\Gs}-saturated for $\rho$.
\end{lemma}

Before we prove Lemma~\ref{lem:sfclos:isat}, let us use it to show that~\eqref{eq:sfclos:ginc} holds and conclude the completeness proof. We need to show that $\bratauxsfc(\gtype{\veps}) \subseteq S$. Hence, we consider $r \in \bratauxsfc(\gtype{\veps})$ and prove that $r \in S$.

By definition, we know that $r \in \opti{\sfp{\Cs}}{\gtype{\veps},\rho}$. Moreover, since $\Gs \subseteq \Cs$, we have $\sfp{\Gs} \subseteq \sfp{\Cs}$. This yields the following inclusion,
\[
\opti{\sfp{\Cs}}{\gtype{\veps},\rho} \subseteq \opti{\sfp{\Gs}}{\gtype{\veps},\rho}
\]
Consequently, we get $r \in \opti{\sfp{\Gs}}{\gtype{\veps},\rho}$ which can be reformulated as follows:
\[
(\gtype{\veps},r) \in \popti{\sfp{\Gs}}{\Gs}{\rho}.
\]
Theorem~\ref{thm:sfclos:carfinite} states that $\popti{\sfp{\Gs}}{\Gs}{\rho}$ is the least $\sfp{\Gs}$-saturated subset of $(\gclac) \times R$. Hence, since $S'$ is also $\sfp{\Gs}$-saturated by Lemma~\ref{lem:sfclos:isat}, we obtain that $(\gtype{\veps},r) \in S'$. By Fact~\ref{fct:sfclos:sp}, this implies that $r \in S$, concluding the main argument. It remains to prove Lemma~\ref{lem:sfclos:isat}.

\begin{proof}[Proof of Lemma~\ref{lem:sfclos:isat}]
	We show that $S' \subseteq (\gclac) \times R$ is $\sfp{\Gs}$-saturated for $\rho$. This involves four properties. They are proved independently.
	
	\medskip
	\noindent
	{\bf Trivial elements.} Consider $w \in A^*$, we have to show that $(\gtype{w},\rho(w)) \in S'$.  By definition of $S'$, it suffices to prove that $\rho(w) \in \dclosr \sfrats(\gtype{w})$. Clearly, $w \in \gtype{w}$. Hence, it now suffices to prove that $\rho(w) \in \sfrats(w)$. If $w = \veps$, then $\rho(w) = 1_R$ and $\sfrats(\veps)= \{1_R\}$ (this is the neutral element of $2^R$). Thus, it is immediate that $\rho(w) \in \sfrats(w)$. Assume now that $w \in A^+$. We have $a_1,\dots,a_n \in A$ such that $w=a_1 \cdots a_n$. Since $S$ is \sfp{\Cs}-complete, we have $1_R \in S$ (this is obtained from \Cs-operation, see Remark~\ref{rem:sfclos:compne}). Hence, we get that for every $i \leq n$, $\rho(a_i) \in S \cdot \{\rho(a_i)\} \cdot S = \sfrats(a_i)$. It then follows that $\rho(w) \in \sfrats(w)$ which concludes this case.
	
	\medskip
	\noindent
	{\bf Downset.} Consider $(C,r) \in S'$ and $q \leq r$. We show that $(C,q) \in S'$. By definition, there are two possible cases. First, it may happen that $C = \gtype{\veps}$ and $r \in S$. In that case, $q \in S$ since $\dclosr S = S$ ($S$ is \sfp{\Cs}-complete) and we get that $(C,q) = (\gtype{\veps},q)\in S'$. Otherwise, $r \in \dclosr \sfrats(C)$ which means that $q \in \dclosr \sfrats(C)$ as well and we get $(C,q) \in S'$, concluding the proof for downset.
	
	\medskip
	\noindent
	{\bf Multiplication.} Consider $(C,q),(D,r) \in S'$. We have to show that $(C \cmult D,qr) \in S'$. By definition of $S'$, there are several cases.
	
	Assume first that $C = D = \gtype{\veps}$. In that case, Fact~\ref{fct:sfclos:sp} yields that $q,r  \in S$. Since $S$ is \sfp{\Cs}-complete, this implies that $qr \in S$ and we obtain that $(C \cmult D,qr) = (\gtype{\veps},qr) \in S'$. We now assume that either $C$ or $D$ is distinct from $\gtype{\veps}$ for the remainder of this case.
	
	Assume first that both $C$ and $D$ are distinct from $\gtype{\veps}$. By definition of $S'$, this implies that $q \in \dclosr\sfrats(C)$ and $r \in \dclosr \sfrats(D)$. It follows that $qr \in \dclosr \sfrats(CD)$. Since $CD \subseteq C \cmult D$, this yields $qr \in  \dclosr \sfrats(C \cmult D)$ and we get $(C \cmult D,qr) \in S'$.
	
	Finally, we handle the case when $C = \gtype{\veps}$ and $D \neq \gtype{\veps}$ (the symmetrical case is left to the reader). By Fact~\ref{fct:sfclos:sp}, the hypothesis that $C = \gtype{\veps}$ yields $q \in S$. Moreover, we have $r \in \dclosr \sfrats(D)$ since $D \neq \gtype{\veps}$. Observe that the hypothesis $D \neq \gtype{\veps}$ also implies that $D \subseteq A^+$. Since $S \cdot S \subseteq S$ ($S$ is \sfp{\Cs}-complete), one may verify from the definition of \sfrats that this yields $S \cdot \sfrats(D) \subseteq \sfrats(D)$. Thus, since $q \in S$ and $r \in \dclosr \sfrats(D)$, we get $qr \in \dclosr \sfrats(D)$. Finally, $C \cmult D = D$ since $C = \gtype{\veps}$. Thus, $qr \in \dclosr \sfrats(C \cmult D)$ which yields $(C \cmult D,qr) \in S'$, concluding the proof for multiplication.
	
	\medskip
	\noindent
	{\bf \sfp{\Gs}-closure.} It remains to handle \sfp{\Gs}-closure. Since \Gs is a sub-class of \Cs, it is a \vari of group languages. Thus, one may verify that \gclac is a group which implies that $\gtype{\veps}$ is the only idempotent of $\gclac$. Hence, it suffices to show that for every $r \in R$ such that $(\gtype{\veps},r) \in S'$, we have $(\gtype{\veps},r^{\omega} + r^{\omega+1}) \in S'$. By Fact~\ref{fct:sfclos:sp}, we have $r \in S$. Thus, since $S$ is \sfp{\Cs}-complete, \sfp{\Cs}-closure yields that $r^{\omega} + r^{\omega+1} \in S$ which implies $(\gtype{\veps},r^{\omega} + r^{\omega+1}) \in S'$ by definition, finishing the proof.	
\end{proof}
 
\end{document}